\documentclass[11pt]{article}
\usepackage[margin=1in,letterpaper]{geometry}
\usepackage{graphicx} 

\newif\ifdraft%
\drafttrue

\usepackage{algorithm}
\usepackage{algpseudocode}
\usepackage{amsthm}
\usepackage{mathtools}
\usepackage{amssymb}
\usepackage{url}
\usepackage{thmtools,thm-restate}
\usepackage[unicode,hidelinks, hypertexnames=false]{hyperref}
\hypersetup{
    colorlinks=true,
    linkcolor=blue,
    citecolor=blue,
    filecolor=black,
    urlcolor=[HTML]{0088CC},
    pdftitle={Near-Optimal Distributed Ruling Sets for Trees and High-Girth Graphs},
    pdfauthor={Anonymous}
}
\usepackage{lipsum}

\newcommand\blfootnote[1]{%
  \begingroup
  \renewcommand\thefootnote{}\footnote{#1}%
  \addtocounter{footnote}{-1}%
  \endgroup
}

\usepackage[capitalize, nameinlink]{cleveref}
\usepackage{xcolor}
\usepackage{xspace}
\usepackage{paralist}
\usepackage{enumitem} 

\usepackage{dsfont}

\newtheorem{remark}{Remark}
\newtheorem{theorem}{Theorem}
\newtheorem{definition}[theorem]{Definition}
\newtheorem{observation}[theorem]{Observation}

\newtheorem{lemma}[theorem]{Lemma}
\newtheorem{claim}[theorem]{Claim}
\title{Near-Optimal Distributed Ruling Sets for Trees\\ and High-Girth Graphs}
\author{Malte Baumecker \\ TU Graz \and Yannic Maus \\ TU Graz  \and  Jara Uitto \\ Aalto University \\  }
\date{}

\newcommand{\poly}{\textrm{poly}}
\newcommand{\local}{\textrm{LOCAL}\xspace}
\newcommand{\CONGEST}{\ensuremath{\mathsf{CONGEST}}\xspace}
\newcommand{\LOCAL}{\ensuremath{\mathsf{LOCAL}}\xspace}
\newcommand{\LMJ}{\ensuremath{\textsc{Local-Minima-Join}}\xspace}
\newcommand{\cleanup}{\ensuremath{\textsc{Clean-Up}}\xspace}
\newcommand{\DThres}{\ensuremath{\Delta_*}}
\newcommand{\Dmax}{\ensuremath{\Delta_{\mathrm{max}}}}
\newcommand{\Dsamp}{\ensuremath{\Delta_{\mathrm{small}}}}
\newcommand{\procdegreedrop}{\textsc{Degree-Drop}\xspace}
\newcommand{\procdegreedropsampling}{\textsc{Degree-Drop-Sampling}\xspace}
\newcommand{\finishoff}{\ensuremath{\textsc{Finish-Off}}\xspace}
\newcommand{\LMJS}{\ensuremath{\textsc{Local-Minima-Join-Sampling}}\xspace}
\newcommand{\LMJshort}{\ensuremath{\textsc{LMJ}}\xspace}
\newcommand{\LMJSshort}{\ensuremath{\textsc{LMJ-Sampling}}\xspace}

\usepackage{todonotes}
\newcommand{\ym}[1]{\ifdraft \todo{\textcolor[rgb]{0.1,0.6,0.1}{Y: #1}}\fi}
\newcommand{\mb}[1]{\ifdraft \todo{\textcolor[rgb]{0.1,0.1,0.6}{M: #1}}\fi}

\begin{document}

\maketitle
\thispagestyle{empty}
    \section*{Abstract}
    Given a graph $G=(V,E)$, a $\beta$-ruling set is a subset $S\subseteq V$ that is i) independent, and ii) every node $v\in V$ has a node of $S$ within distance $\beta$.  In this paper, we present almost optimal distributed algorithms for finding ruling sets in trees and high girth graphs in the classic \local model. 
    As our first contribution, we present an $O(\log\log n)$-round randomized algorithm for computing $2$-ruling sets on trees, almost matching the $\Omega(\log\log n/\log\log\log n)$ lower bound given by Balliu et al.\ [FOCS'20]. Second, we show that $2$-ruling sets can be solved in $\widetilde{O}(\log^{5/3}\log n)$ rounds in high-girth graphs. Lastly, we show that $O(\log\log\log n)$-ruling sets can be computed in $\widetilde{O}(\log\log n)$ rounds in high-girth graphs, matching the lower bound up to triple-log factors. All of these results either improve polynomially or exponentially on the previously best algorithms and use a smaller domination distance $\beta$. 

\blfootnote{This research was funded in whole or in part by the Austrian Science Fund (FWF) \url{https://doi.org/10.55776/P36280}, \url{https://doi.org/10.55776/I6915}. For open access purposes, the author has applied a CC BY public copyright license to any author-accepted manuscript version arising from this submission.}
    \tableofcontents
    \clearpage
\clearpage
    \setcounter{page}{1}
    \section{Introduction}

    A $\beta$-ruling set (RS) is a set of pairwise non-adjacent nodes such that any node in an input graph $G = (V, E)$ is within $\beta$ hops of some ruling set node.
    The Maximal Independent Set (MIS) problem, that is a $1$-ruling set, is among the most extensively studied problems in distributed graph algorithms.
    In many cases, ruling sets for $\beta \geq 2$ can be faster to compute than maximal independent sets and have applications for various problems including, $\Delta$-coloring~\cite{GHKM18} (for maximum degree $\Delta$), network decompositions~\cite{GG24,awerbuch89}, and MIS~\cite{GhaffariImproved16}. 
    Ruling sets are also interesting in their own right and have been extensively studied in various settings, e.g., ~\cite{schneider2013, Balliu2020-ruling, Gfeller07,CONGEST_rulingsets, GG24, Pai2022, kothapalli-superfast2012, cambusCCruling, BEPSv3, schneider2010, assadi-streaming-ruling}.

For MIS, a classic result states that we cannot beat $\Omega(\sqrt{\log n/\log \log n})$ rounds and the bound holds even on trees \cite{kuhn16_jacm}. For ruling sets, the strongest lower bound is the minimum of $\widetilde{\Omega}(\log^{1/(\beta + 1)} \Delta)$ and $\widetilde{\Omega}(\log_\Delta \log n)$ \cite{BBKO2021hideandseek}, which does not rule out, for example, double-logarithmic runtimes.
    The main question that we partially answer in this paper is how close to optimal can we get for ruling sets in terms of both the domination distance $\beta$ and the runtime.
    Next, we present our results and then further background in detail.


    \subsection{Our Contributions}
    We consider the \local model of distributed computing introduced by Linial~\cite{linial92}.
    A communication network is modeled as a graph $G = (V, E)$, where $|V| = n$ computational units communicate over the edges in synchronous message-passing rounds.
    Message sizes are unbounded and each unit can perform arbitrary local computations. 
    In the beginning, each node only knows its neighbors, and at the end of a computation, each node $v$ should know its own output, for example, whether it is a part of the ruling set or not. 
    The runtime of an algorithm is measured in the number of rounds.
    
    First, we give an algorithm for a $2$-ruling set on trees in $O(\log \log n)$ time with the best-possible domination distance $\beta = 2$. 

    \begin{restatable}{theorem}{thmRulingSetTrees}
    \label{thm:rulingSetTrees}
        There is an $O(\log\log n)$-round algorithm to compute a $2$-ruling set in trees w.h.p.\
    \end{restatable}
    Our algorithm is optimal up to a triple-logarithmic factor in the runtime and it improves on the previous $\widetilde{O}(\log^2 \log n)$ time algorithm for trees that only achieves a $3$-ruling set~\cite{kothapalli-superfast2012}.
    Second, we introduce an algorithm for graphs of girth at least $7$ that obtains a $2$-ruling set and runs in $\widetilde{O}(\log^{(5/3)} \log n)$ rounds. We note that this algorithm is also very close to optimal: the domination distance cannot be improved in this runtime domain, and the runtime is within a sub-double-logarithmic factor from optimal.
    

    \begin{restatable}{theorem}{thmRulingSetHighGirth}
    \label{thm:rulingSetHighGirth}
        For graphs of girth at least $7$, there is a randomized algorithm to compute a $2$-ruling set in $\widetilde{O}(\log^{5/3}\log n)$ rounds w.h.p.
    \end{restatable}
    This result improves the previous best algorithm for high girth graphs, which, similar as on trees, obtains a $3$-ruling set in time $O(\log^3 \log n)$~\cite{kothapalli-superfast2012}.
        Finally, by slightly relaxing $\beta$ to $O(\log \log \log n)$, we give an algorithm for high-girth graphs whose runtime is within a triple-logarithmic factor from the corresponding lower bound and exponentially improves on the current best algorithm for general graphs that obtains an $O(\log \log n)$-ruling set in $O(\log \log n)$ time.
    
    \begin{restatable}{theorem}{thmlogloglog}
    \label{thm:logloglog}
        For graphs of girth at least 7, there is a randomized algorithm that w.h.p.\ computes an $O(\log \log \log n)$-ruling set in $\widetilde{O}(\log\log n)$ rounds.
    \end{restatable}

By combining the core technical contributions (different fast degree reduction procedures) of \Cref{thm:rulingSetTrees,thm:rulingSetHighGirth} algorithms with the known sublogarithmic deterministic MIS algorithm for trees by Barenboim and Elkin \cite{Barenboim2010}, we also obtain the following corollary.  
\begin{restatable}{corollary}{corruling}
\label{cor:2rulingtreeV2}
There is a randomized algorithm that w.h.p.\ computes a $2$-ruling set of a tree in $O(\log\log \Delta)+O(\log\log n/\log\log \log n)$ rounds. 
\end{restatable}
This is optimal as long as $O(\log\log\Delta)=O(\log\log n/\log\log\log n)$ and near-optimal otherwise. To the best of our knowledge when parameterizing the runtime as a function of $n$ and $\Delta$ the best prior $2$-ruling set algorithm on trees requires 
$O(\log^{1/3}\Delta+\poly\log\log n)$ rounds and can be obtained by combining \cite{BEPSv3} and \cite{GhaffariImproved16}.

    \paragraph{Revisiting MIS on Trees.}
    A key part of the current state-of-the-art result for MIS on trees is to give stronger guarantees to the classic algorithms (such as Luby's) using neighborhood independence~\cite{lenzen-tree-2011}.
    Our core technical ingredient for \Cref{thm:rulingSetTrees} is a better analysis of Luby's algorithm when run on trees, and, in particular, we aim to leverage the independence in a similar fashion as \cite{lenzen-tree-2011}.
    While their general results are correct, there are some dependencies overlooked in the proof of one of the central lemmas of \cite{lenzen-tree-2011}. As a part of our work, we show how these dependencies can be formally handled.

\subsection{Technical Overview}
\label{sec:tecOverview}
The basic building block of our results is a natural variant of the classic MIS algorithm designed independently by Luby and Alon, Babai, and Itai~\cite{luby86, alon86}. 
In rounds, each node picks a random number from the range $[0, 1]$ and becomes a ruling set node if it is a local minimum.
Then, each node adjacent to a local minimum is removed from the graph.
The classic analysis for MIS in general graphs shows that in each round, a constant fraction of edges disappear from the graph in expectation, which leads to an $O(\log n)$ runtime.
Later, Lenzen and Wattenhofer showed that this process terminates much faster in trees,i.e., in $\widetilde{O}(\sqrt{\log n})$ rounds \cite{lenzen-tree-2011}.
They used a different line of attack and showed that within $\widetilde{O}(\sqrt{\log n})$ rounds, the degree of each node becomes small or they are removed from the graph.

\paragraph{Our Analysis: Take I.} 
To obtain a double-logarithmic runtime, our goal is to reduce the maximum degree $\Delta_i$ in iteration $i$ to $\Delta_{i+1} \leq \Delta_{i}^{1 - \varepsilon}$ for iteration $i + 1$. 
The first simple observation is that since we are looking for a 2-ruling set, we can always remove the 2-hop neighborhoods of the local minima from the graph. Second, we do not aim for a degree reduction that holds with high probability, but rather we are only aiming that the probability for a node to retain a large degree is at most $1/\poly\Delta$, which is sufficient to use the powerful shattering framework. Observe that a node with degree $d$ becomes a local minimum with probability $1/d$.
The analysis is then split into two cases for each node $v$. 
Our goal is to show that there are no nodes with degree more than, say $\Delta_i^{0.8}$, after iteration $i$.
So, consider a node $v$ with degree more than $\Delta_i^{0.8}$ and notice that $v$ falls into at least one of the following cases that can be analyzed easily when \emph{ignoring} dependencies.
\begin{compactenum}[1)]
\item  If there are at least $\Delta^{0.6}$ neighbors of degree at most $\sqrt{\Delta}$, the probability that no neighbor is  a local minimum is at most $(1-1/\sqrt{\Delta})^{\Delta^{0.6}}\leq \exp(-\Delta^{-0.1})$, or
\item there are at least $\Delta^{1.1}$ unique $2$-hop neighbors. Then, none of the $2$-hop neighbors is a local minimum with probability at most $(1-1/\Delta)^{\Delta^{1.1}} \leq \exp(-\Delta^{-0.1})$. 
\end{compactenum}
Hence, in both cases, a node with degree more than $\Delta_i^{0.8}$ in the beginning of iteration $i$ will be removed from the graph with probability at least $1-\exp(-\Delta^{-0.1})$. 
This exponentially small error probability would be sufficient to use the shattering framework \cite{BEPSv3}, by placing nodes failing the progress guarantee into a clean-up phase that can be dealt with at the very end of the algorithm. 

\paragraph{Main Challenge: Dependencies.} On a quick glance, the previous paragraph may seem to give the desired intuition for a doubly-exponential degree drop. 
Unfortunately, the events whether nodes in $v$'s $1$-hop neighborhood are local minima are not independent of  $v$'s random choice and hence not mutually independent. Similarly, these events are not independent for all nodes in the $2$-hop neighborhood. 
In fact, with probability at least $\Theta(1/d(v))$, the random number of node $v$ is one of the two smallest random numbers in its $1$-hop neighborhood. Let $u$ be the other node whose random number is among the two smallest numbers in $v$'s $1$-hop neighborhood.
Conditioned on the event that $v$'s random number is one of the two smallest random numbers in its $1$-hop neighborhood, the probability that $v$ is the local minimum is at most $1/2$. We are also not guaranteed that $u$ is a local minimum because $u$ may have other neighbors. Even conditioned on the event that the random number of $u$ is smaller than all of the random numbers of $u$'s $1$-hop neighbors except the random number of $v$, there is a constant probability that $u$ is not a local minimum itself. 
Hence, we cannot hope that $v$ is removed with probability higher than (roughly) $1 - \Theta(1/d(v))$, contradicting initial intuition.

\paragraph{Our Analysis: Take II.} Our analysis still follows the intuition of the two cases above but the analysis is more complex for each case. As explained, for the first case, the probability to remain in the graph can only be bounded by roughly $\Theta(1/d(v))$. In order to analyze this probability, we root the tree at $v$. We call a node \emph{successful} if it is a local minimum with regard to its children in the tree, but ignoring its parent. The advantage is that the event of being successful is independent for different nodes on the same level of the tree. For case 1) we show that $v$ is likely to have  $\Delta^{0.1}$ successful neighbors. Now, following the counter example above, it is not unlikely that $v$ actually \emph{eliminates} all its successful neighbors by having a smaller random value than all of them. But we bound the probability that this happens by $O(1/\Delta^{0.6})$, see \Cref{lem:preparation}. To get a probability that is inversely proportional to the number of low-degree neighbors---as opposed to the (expected) number of successful nodes---we take the skew on the probability distribution of successful nodes into account, i.e., once we condition on a node being successful, its probability distribution in $[0,1]$ is not uniform anymore. 
For case 2) we perform a similar analysis but need to reason about successful nodes in the $2$-hop neighborhood that are not eliminated by their $1$-hop neighbors.

As a side result of this analysis, we also recover one of the central progress lemmas in the analysis by Lenzen and Wattenhofer \cite{lenzen-tree-2011}. This is interesting, as to the best of our understanding, their analysis overlooked some dependencies. They also root the tree at node $v$ and have a similar definition of a node being successful. But once they condition on a node being successful, they reason about the random values of the node's children and assume their independence and uniform distribution in $[0,1]$, both of which seem to be faulty, see \Cref{rem:lenzen} for our resolution.

\paragraph{The Clean-up.}
As mentioned, our degree reduction relies on the shattering framework, that is, in each iteration we obtain a leftover set of nodes that we need to deal with at the end of the process. Each of these leftover sets consists of connected components of size at most $N=O(\poly\Delta_i \log n)$. Optimally, we would just like to run a simple deterministic logarithmic time MIS or ruling set algorithm on them, but as $\Delta_i$ may be very large this is not sufficient for our runtime goal of $O(\log\log n)$. Instead, inspired by \cite{CHLPU20}, we prove that additional structure exists in the small components. This can be used to show that a classic rake \& compress-based procedure \cite{Miller1985} can deal with each connected component in $O(\log\log n)$ time. In order to process all $O(\log\log n)$ leftover sets in parallel, we also remove the neighborhood of each set before going to the next iteration of the degree reduction. Hence, there are no edges between the sets, and they can be handled in parallel.

\paragraph{High-girth graphs.}
The analysis of our probabilistic degree reduction does not require that the graph actually is a tree, but having high girth is sufficient, except that the rake \& compress procedure in the clean-up phase does not run in $O(\log\log n)$ rounds. In fact, a rake \& compress-based approach that relies on iteratively removing nodes with constant degrees does not even make sense in such graphs. As components are of size $N=\poly\Delta\log n\gg \poly\log n$ we cannot even profit from the very efficient $\widetilde{O}(\log^{5/3} N)$ MIS algorithm from \cite{GG24}. 

Hence, we design a second-degree drop procedure that reduces the degrees to $\Delta'=\poly\log n$ in $O(\log\log n)$ rounds. The algorithm enhances the previous one by incorporating a sampling-based approach and introducing distinct phases within each iteration of the degree reduction process, guaranteeing the degree drop with high probability and without producing any leftover components. When degrees are polylogarithmic in $n$ we can revert to our previous method and use the shattering framework as components are now bounded by $N=O(\poly\Delta'\log n)=O(\poly\log n)$. Finishing them off with the MIS algorithm of \cite{GG24} takes $\widetilde{O}(\log^{5/3} N)=\widetilde{O}(\log^{5/3}\log n)$ rounds and yields \Cref{thm:rulingSetHighGirth}. For \Cref{thm:logloglog} we replace the MIS computation on the small components with the more efficient algorithm of \cite{GG24} to compute $O(\log \log N)=O(\log \log \log n)$-ruling sets in time $\widetilde{O}(\log N)=\widetilde{O}(\log\log n)$.  

In principle this second degree reduction is better than our first degree reduction, also on trees, and it could also be extended to work for graphs with smaller than polylogarithmic degrees, but still requiring the shattering framework. Still, we believe that our analysis, the result, and recovering Lenzen's \& Wattenhofers progress lemma is interesting in its own right.  



\subsection{Related Work}
We have already covered the most important related work. In addition to the results presented in this section, we also recommend the excellent related work section and overview tables in the lower bound papers \cite{BBKO2021hideandseek,BBO22} and in \cite{BEPSv3} for further reference. 

\paragraph{Lower bounds.} As mentioned, the strongest lower bound for computing $\beta$-ruling sets is due to \cite{BBKO2021hideandseek} and proves that randomized algorithms requires the minimum of $\widetilde{\Omega}(\log^{1/(\beta + 1)} \Delta)$ and $\widetilde{\Omega}(\log_\Delta \log n)$ rounds while deterministic algorithms require at least $\Omega(\beta\Delta^{1/(\beta)})$ and $\Omega(\log n/(\beta\log\log n))$ rounds. Much earlier it was known that MIS, i.e., $1$-ruling sets require $\Omega(\log^*n)$ rounds. Linial showed this for deterministic algorithms \cite{linial92} and Naor for randomized algorithms \cite{Naor91}. Both bounds also hold for $\beta$-ruling sets whenever $\beta$ is constant. 

\paragraph{Randomized Algorithms. } Ghaffari's seminal MIS algorithm is still the state of the art for MIS on general graphs and runs in $O(\log\Delta)+O(\poly\log\log n)$ rounds \cite{GhaffariImproved16}. His algorithm can also be combined with prior work on ruling sets \cite{tushar-super-fast-2014} to compute $\beta$-ruling sets in $O(\beta\log^{1/\beta}\Delta)+O(\poly\log\log n)$ rounds.

Kothapalli and Pemmaraju showed that one can obtain a sublogarithmic, $O(\log^{3/4} n)$ round algorithm for 2-ruling sets, and for 3-ruling sets, a $\widetilde{O}(\log^2 \log n)$-rounds algorithm for trees and a $O(\log^3 \log n)$-rounds algorithm for bounded arboricity and girth at least 7; the runtime stated in the paper is worse, but it can be improved by plugging in the newest results on network decompositions \cite{GG24}  \cite{kothapalli-superfast2012}. Bisht, Kothapalli and Pemmaraju generalized the $3$-ruling set approach to $\beta$-ruling sets for general graphs; we have already discussed that it can be combined with \cite{GhaffariImproved16} to obtain $\beta$-ruling sets in $O(\beta\log^{1/\beta}\Delta)+O(\poly\log\log n)$ rounds \cite{tushar-super-fast-2014}. 
Gfeler and Vicari, combined with \cite{BEPSv3}  obtain $O(\log \log n)$-ruling set in time $O(\log \log n)$ plus runtime of MIS on graphs with polylogarithmic degrees \cite{Gfeller07}.

\paragraph{Deterministic algorithms.}
Early works computed $O(\log n)$-ruling sets deterministically in $O(\log n)$ rounds \cite{awerbuch89}. Also see \cite{CONGEST_rulingsets,KMW18-rulingSets} who slightly improved this algorithm to work without unique IDs and to trade runtime for domination parameter $\beta$. There are also algorithms that optimize the truly local complexity of ruling sets with the current state of the art by Maus for $\beta$-ruling sets being $O(\Delta^{2/(\beta+2)})+O(\log^* n)$ rounds \cite{M21}, slightly improving a result of a similar flavor of Schneider, Elkin, and Wattenhofer \cite{schneider2013}.

Barenboim and Elkin designed a deterministic MIS algorithm (and hence also one for ruling sets) for trees and graphs of bounded arboricity. It has runtime  $O(\log n/\log\log n)$ \cite{Barenboim2010}.  Recently \cite{GG24} obtained MIS in $\widetilde{O}(\log^{5/3}n)$ deterministically, and $O(\log \log n)$-ruling sets in $\widetilde{O}(\log n)$ rounds, also deterministically. 

For ruling set algorithms in the \CONGEST model we refer to \cite{MPU23}, and for the congested clique to \cite{CKPU23} and for the massively parallel computation model to \cite{GP24}. 


\section{Preliminaries \& Notation}

    Given a graph $G=(V,E)$, we denote by $N(v)$ the neighborhood of $v$ (excluding $v$ itself). For a node $v\in V$ we let $d_v \coloneqq |N(v)|$ denote the degree of $v$ and for two nodes $v,w$ we let $dist_G(v,w)$ be the length of the shortest path between $v$ and $w$ in $G$. Furthermore we denote by $N_2(v)$ all nodes $w \in V$ with $dist_G(v,w)\leq2$, i.e., the 2-hop neighborhood of $v$ (again excluding $v$ itself). By $N_2^-(v)$ we denote all nodes $w \in V$ with $dist_G(v,w)=2$, i.e., the \emph{exclusive} 2-hop neighborhood of $v$. Note that $N_2^-(v) \subset N_2(v)$ and $N_2^-(v) = N_2(v) \setminus N(v)$. For a subset $W \subset V$, let $G[W]$ be the induced subgraph of the subset $W$. 

    For given numbers $r_v$ for each node $v$, a node $u$ is a \emph{local minimum} if its number $r_u$ is smaller than the number of all of its $1$-hop neighbors, i.e., $r_u < r_w$ for all $w\in N(u)\setminus\{u\}$. 
    
    A set $U\subseteq V$ of nodes is \emph{$x$-independent} for some integer $x\geq1 $ if $dist(u,v)>x$ for all nodes $u,v \in U$.

     An algorithm is correct w.h.p. (with high probability) if there exists a constant $c>1$ such that it errs with probability $\leq 1/n^c$. 

     We use the notation $\widetilde{O}$ as follows: $g(n) \in \widetilde{O}(f(n))$ if there exists a constant $c$ such that $g(n) \in O(f(n) \cdot \log^c(f(n)))$.

    \section{Ruling Sets in Trees}
    \label{sec:trees}
    The goal of this section is to prove the following theorem. 
\thmRulingSetTrees*

As explained in \Cref{sec:tecOverview}, the core of  our result is the following lemma. It shows that a constant number of iterations of the \emph{pick a random number, let local minima join the ruling set, and remove the $2$-hop neighborhood} is likely to drastically reduce the maximum degree induced by uncovered nodes. Degree reduction does not hold with high probability, but nodes not satisfying the claim will in some sense form small components (see \ref{itm:4}) that can be solved efficiently thereafter. 

    \begin{restatable}[Degree-drop Lemma]{lemma}{lemDegreeDropTrees}
    \label{lem:degreedroptrees}
    Let $c>0$ be a constant.
    For a tree $T$ with maximum degree $\Delta \geq 18$ there is a $O(1)$-round \LOCAL algorithm that finds sets of nodes $S,W\subseteq V(T)$ such that with probability $1-n^{-c}$ the following hold,
 
   \begin{compactenum}[(a)]
          \item $S$ is an independent set\label{itm:1},
          \item There is no edge between any node in  $W$ and any node in $S$.\label{itm:2}, 
          \item the maximum degree  of $G[V(T)\setminus (S \cup N_2(S)\cup W)]$ is at most $\DThres=\Delta^{3/4}$\label{itm:3},
          \item each connected component of $T[W]$ has a 6-distance dominating set of size $O(\log n)$. The distance $6$ is measured in $T$ and not in $T[W]$.
          \label{itm:4}
        \end{compactenum}   
        All of these properties except for \ref{itm:4} hold deterministically regardless of the choice of the algorithms random bits. 
    \end{restatable}
    
  The procedure of \Cref{lem:degreedroptrees} is denoted by \procdegreedrop. We detail on it in \Cref{sec:degreeDrop} where we also present pseudocode. Next, we continue with the high level overview of the algorithm. 

\paragraph{High level overview on Ruling Set Algorithm:}    We iterate the procedure \procdegreedrop from \Cref{lem:degreedroptrees} $R = O(\log \log n)$ times, resulting in independent sets $S_1,\dots,S_R$, unsolved put-aside sets $W_1, \ldots W_R$, and a remaining forest with a constant maximum degree.  In iteration $i$, we also remove $N(W_i)$ from the graph, ensuring that $W_i$'s are non-adjacent. \Cref{lem:degreedroptrees}  shows that each of the $W_i$'s can be covered by an $O(1)$-distance dominating set of size $O(\log n)$, i.e., a subset of nodes $D_i\subseteq W_i$ such that any $w\in W_i$ has a node in $D_i$ in in constant distance. This dominating set is never actually constructed but its existence can be used to prove that a certain rake-and-compress based algorithm to compute an MIS $Z_i$ of $T[W_i]$ runs in  $O(\log\log n)$ rounds \cite{CHLPU20}. As $W_i$s are non-adjacent each of the $R$ $W_i$s can be handled in parallel in this clean-up phase. See \Cref{sec:cleanup} for the details. Lastly, we find an MIS $S_{R+1}$ on the remaining constant-degree instance in $O(\log^*n)$ rounds \cite{linial92}. In total we return the union of the computing sets $S_1, \ldots, S_{R+1}$, and $Z_1, \ldots, Z_R$. The result is a $2$-ruling set. 
 Next, we provide pseudocode for the whole ruling set algorithm. The details of each step are presented in the sections thereafter.

      \begin{algorithm}[!htbp]
    \caption{Randomized 2-ruling set for trees}
    \label{alg:complete}
    \begin{algorithmic}[1]
       \State Initialize $S_i,W_i,Z_i \leftarrow \emptyset$ for $i=1,\dots,R$ and $S_{R+1} \leftarrow \emptyset$.
       \For{$i=1, \dots, R=O(\log\log n)$} \label{line:treesforloop}
       \State $S_i,W_i \leftarrow$ \procdegreedrop($T$, $\Dmax^{(3/4)^i}$)
       \State Remove  $S_i \cup N_2(S_i) \cup W_i\cup N(W_i)$ from $T$ \label{alg:complete:neighborremoval}
       \EndFor\label{line:endtreeforloop}
        \State \textbf{for} $i=1,\dots,R$ in \textbf{parallel}: $Z_i \leftarrow $CleanUp($W_i$) \label{alg:complete:cleanup}
       \State $S_{R+1} \leftarrow $MIS(T)  \emph{//using $O(\log^*n)$ rounds on remaining graph with $O(1)$ maximum degree  \cite{linial92}}\label{alg:complete:finalMIS}
       \State \textbf{return:} $S_{R+1}\cup \bigcup_{i=1}^R (S_i \cup Z_i)$
    \end{algorithmic}
    \end{algorithm}

\subsection{Degree Drop}    
\label{sec:degreeDrop}
The objective of this section is to prove \Cref{lem:degreedroptrees}. 

     Through this section, we will use $\Dmax$ to denote the maximum degree of the input tree of $\procdegreedrop(T,\DThres)$ and $\DThres$ to denote the threshold degree to which we want to drop in one call of \procdegreedrop. For our purposes we let $\DThres \coloneqq \Dmax^{3/4}$. In particular in \Cref{alg:complete} we invoke \procdegreedrop($T,\Dmax^{(3/4)^i}$) in the $i$-th iteration, ensuring the maximum degree drops to the respective threshold. 

    \procdegreedrop (details below) begins with an empty ruling set $S=\emptyset$. Then we iteratively for $16c$ iterations let uncovered nodes pick a  numbers in $[0,1]$ uniformly at random, local $1$-hop minima join the ruling set $S$, and their $2$-hop neighbors are removed from the graph. The constant $c$ is the constant used in \Cref{lem:degreedroptrees} and we remark that the probability in that statement is tuneable. This $O(1)$-round procedure that we term \LMJ has been used frequently in the literature to compute ruling sets and maximal independent sets on trees and general graphs, e.g., \cite{kothapalli-superfast2012,tushar-super-fast-2014,lenzen-tree-2011,luby86}. See \Cref{alg:lmj} for pseudocode. 
    \begin{algorithm}[!htbp]
    \label{alg:lmj}
      \begin{algorithmic}[1]
        \Procedure{\LMJ}{$T$}  \emph{// we use $\LMJshort $ as short notation}
        \State $S \leftarrow \emptyset$
        \State Uniformly and independent at random compute a real $r_v \in [0,1]$ for all $v\in V(T)$
        \State In parallel for all $v \in V(T)$ 
        \If{$r_v < r_w \forall w \in N(v)$}
        \State $S \leftarrow S \cup \{v\}$ \emph{// those are the nodes joining the independent set}
        \EndIf
        \State Return $S$
        \EndProcedure
      \end{algorithmic}
    \end{algorithm}

    
    \begin{algorithm}
    
      \begin{algorithmic}[1]
      
        \Procedure{\procdegreedrop}{$T, \DThres$}  
        \label{alg:degdrop}
        \State $S \leftarrow \emptyset$
          \For{$j=1,\dots, 16c$}
            \State $S \leftarrow$ $S\cup\LMJ(T)$ 
            \State Remove  $S\cup N_2(S)$ from $T$\label{alg:degreeDrop:neighborremoval} \emph{// those are the covered nodes by S}
          \EndFor
          \State $W \leftarrow \{v \in V(T) \mid d_v > \DThres \}$ \emph{// those are the nodes that will treated in the Clean-up \ref{sec:cleanup}}
          \State Return $S,W$
        \EndProcedure
      \end{algorithmic}
    \end{algorithm}





\paragraph{Outline.} In \Cref{sec:localSuccess} we first focus on the local success probability, i.e., the probability that a large-degree node is removed in a single iteration of \LMJ. In \Cref{sec:shattering}, we prove that $16c$ iterations of \LMJ shatters the remaining large-degree nodes into small components. Then in \Cref{sec:lemDegreeDropProof} we prove \Cref{lem:degreedroptrees}. We use $\LMJshort$ as a short notation for \LMJ.

    \subsubsection{Local success probability of \LMJ}
    \label{sec:localSuccess}
The objective of this section is to prove the following lemma.
\begin{lemma}
\label{lem:localSuccess}
 Let $\Dmax$ be the maximum degree of the input tree $T$.
        Consider a node $v \in V(T)$ with degree $d_v \geq \DThres$ at the start of a call of procedure $\LMJshort(T)$ of $\procdegreedrop(T,\DThres)$. Then the probability that $v$ is not covered after the call of $\LMJshort(T)$ is at most $(1/ \Dmax)^{1/16}$.
\end{lemma}

In one case of our analysis we have $k$ $1$-hop neighbors (for some suitable choice of $k$) and each of them is a local minimum with probability $1/\sqrt{k}$. In this case one is tempted to claim that the probability that none of them is a local minimum is at most $(1-1/\sqrt{k})^k\leq e^{-\sqrt{k}}$. 
The main obstacle in proving \Cref{lem:localSuccess}  is that this is not correct as the events whether the neighbors of $v$ are local minima are not independent despite their random choices being independent and despite being in a tree. The same holds for the local-minima-events of $2$-hop neighbors that share a $1$-hop neighbor of $v$. These dependencies have been partially overlooked in the analysis of \LMJshort to design efficient MIS algorithms for trees in \cite{lenzen-tree-2011}.

  \paragraph{Rooting the tree, children, successful nodes, etc.}  Solely for the purpose of analysis, when we talk about a high-degree node $v$ getting removed during the algorithm, we will \emph{root} the tree $T$ at $v$. This allows us to talk about \emph{children} of nodes in that rooted tree. A node $u$ is the \emph{child} of a node $w$ if $dist(u,v) = dist(w,u)+1$. In particular $N(v)$ are the children of $v$ and for every child of $v$ its children are a subset of $N_2^-(v)$.     
    Further we will call a node $w\in N_2(v)$ \emph{successful} during one call of \LMJshort($T$) if $r_w < r_u$ for all children $u$ of $w$ with regard to rooted tree at $v$. Observe that a successful nodes does not 
    need to be a local minimum. 

    \begin{observation}\label{obs:successful}
    A node $w\in N_2(v)$ is successful with probability $1/\deg(w)$.
    \end{observation}

     When talking about a 2-ruling set $S$, we say that a node $v$ is \emph{covered} either if $v \in S$ or $dist(v,S) \leq 2$. This implies if a set of nodes is covered, we can remove them from our instance and compute a $2$-ruling set on the remaining instance. In particular if an independent set $S$ covers all nodes, $S$ is a $2$-ruling set. Finally we introduce the set $U \coloneqq \{v \in V(T) \mid dist(v,S)>3 \}$ of uncovered nodes at the end of the execution of \procdegreedrop.

    
    \paragraph{High level proof idea.} For the analysis we aim to prove that the probability of a node  $v$ with degree $d_v \geq \DThres$ remaining in the graph is small. Bounding this probability is done via  two cases. First we consider the case where $v$  has many low degree neighbors. We show that $v$ is likely to be covered due to a local minimum in its 1-hop neighborhood. As each low-degree $1$-hop neighbor has a good probability to be successful, $v$ is only not removed if the random number $r_v$ of $v$ is smaller than the numbers of all successful $1$-hop neighbors. We show that this probability is small. 
    In the second case node $v$ has many high degree neighbors. We argue that it is likely that $v$ gets covered due to a local minimum in its 2-hop neighborhood $N_2^-(v)$. While we cannot guarantee that individual $2$-hop neighbors are very likely to become successful because they may have a huge degree themselves, but as there are many of them, we have many candidates that may become local minima, also providing a small probability for $v$ to remain uncovered in this case. We summarize the discussion with the following lemma. Intuitively it is simple, but its statement and its proof are overly technical; conditioning has to be done in the right order to avoid dependencies. 
    \begin{restatable}{lemma}{lemPreparation}
    \label{lem:preparation}
    Let $\Dmax$ be the maximum degree of the input tree $T$.
Consider an arbitrary  node $v\in V$ that is active at the beginning of a call of $\LMJshort(T)$ and let $U$ be the set of uncovered nodes after the call. Let $d\geq 0$ be an integer and let
\begin{align*}
    H(v)=\{u\in N(v)\mid \deg(u)> d\} & & L(v)=\{u\in N(v)\mid \deg(u)\leq d\}
\end{align*}
 be the sets of low and high neighbors of $v$,  respectively.  Let $L_v$ be the random variable counting the successful nodes in $L(v)$, and for $u\in H(v)$ let $C_u$ be the number of successful children of $u$.   Then the following holds for every integer $\ell>0$.


  \begin{align}
\mathbb{P}[v \in U] 
           &  \leq \mathbb{P}[L_v<\ell] +\frac{1}{d\cdot \ell+1} 
            & (\textbf{covered by $1$-hop neighbor})
            \label{lem:caseSmallDegreePreparation} \\
\mathbb{P}[v \in U] 
 &  \leq   \prod_{u\in H(v)}\left(\mathbb{P}[C_u<\ell] +\frac{1}{\Dmax\cdot \ell+1}\right) & (\textbf{covered by $2$-hop neighbor})
          \label{lem:caseHighDegreePreparation}
        \end{align}
    \end{restatable}

\begin{proof}[Proof Sketch of \Cref{lem:preparation}]
    For the full proof, see \Cref{app:probability}. 
    For the proof of \Cref{lem:caseSmallDegreePreparation} we only argue about the probability that a node in the direct neighborhood $N(v)$ joins $S$. 
    For the proof of \Cref{lem:caseHighDegreePreparation} we argue only about the probability that a node in the exclusive 2-hop neighborhood $N_2^-(v)$ joins $S$. 
    Here we crucially exploit, that we do not have cycles in the 3-hop neighborhood, i.e., no cycles of length at most $6$. This fact guarantees, that all the children of two distinct nodes in $N_2^-(v)$ are different and thus their probability to be successful are independent.
    Nevertheless children of a node $v$ have the same parent $v$ and thus are not local minima, independently.

    We circumvent this dependency by never conditioning on any event involving the random value of a parent of the considered nodes. For example in the core of the proof for \Cref{lem:caseSmallDegreePreparation} we condition on sets of $1$-hop neighbors of $v$ being successful. Assume that nodes in $W\subseteq N(v)$ are successful. To bound the probability of $v$ being covered we analyze the probability of $\mathbb{P}\big[r_v< \min_{w\in W} r_w \mid W\text{ successful}\big]$. To obtain the bound of the lemma statement we observe that the conditioning skews the probability distribution of $r_w$ favoring smaller numbers. This makes sense since a successful node is more likely to have a small number. 
\end{proof}

    \begin{proof}[Proof of \Cref{lem:localSuccess}]
        Let $v$ be a node that is still in $T$ at the beginning of an iteration (for loop) of $\procdegreedrop$($T,\DThres$) with $d_v > \DThres$. As in \Cref{lem:preparation} let $U$ be the set of uncovered nodes after the call of $\LMJshort$ of that iteration.  Let $d \coloneqq \DThres^{3/4}$ and let $H(v) = \{ u \in N(v) \mid \deg(u) > d\}$ and $L(v) = \{ u \in N(v) \mid \deg (u) \leq d \}$ like in \Cref{lem:preparation}. Observe that $|H(v)|+|L(v)| = d_v \geq \DThres$ holds. 

        
\medskip

        \noindent\textbf{\boldmath Case $|H(v)| \geq \DThres^{3/4}$:} 
        For each $u \in H(v)$ let $C_u$ be the random variable that counts the number of successful children of $u$. 
        When considering only local minima in $N_2^-(v)$, we can assume that every node $u \in H(v)$ has degree $d+1$ and that every node $w\in N_2^-(v)$ has degree $\Dmax$. These assumptions can be made, since they only decrease the probability that there exists a local minimum in $N_2^-(v)$; \Cref{lem:caseHighDegreePreparation} of \Cref{lem:preparation}  only argues about the 2-hop neighborhood $N_2^-(v)$.
        
        For every node $u\in H(v)$, we obtain $\mathbb{E}[C_u] = \frac{d+1}{\Dmax} $ via \Cref{obs:successful}. Since each child of $u$ is successful independently of all the other children, a Chernoff bound yields 
        \begin{align}
            \mathbb{P}[C_u < \mathbb{E}[C_u]/2] < \exp(-\mathbb{E}[C_u]/8) < \exp(-(d+1)/ 8(\Dmax)) < \exp(-d/8 \Dmax). 
        \end{align}
        Now let $\ell = \mathbb{E}[C_u]/2 = (d+1)/2(\Dmax)$ and apply \Cref{lem:caseHighDegreePreparation} of \Cref{lem:preparation} to obtain:
        \begin{align*}
            \mathbb{P}[v \in U] & \leq \prod_{u\in H(v)}\left(\mathbb{P}[C_u<\ell] +\frac{1}{\Dmax\cdot \frac{d+1}{2\Dmax}+1}\right) \\
           &  \leq \prod_{u\in H(v)} \left(  \exp(-d/8\Dmax)+\frac{2}{d}  \right) 
           \leq \left(  \exp(-d/8\Dmax)+\frac{2}{d}  \right)^{|H(v)|}\\
            & \leq \left(\frac{3}{\DThres^{3/4}}\right)^{\DThres^{3/4}} = \left(\frac{3}{\Dmax^{9/16}}\right)^{\DThres^{3/4}}< \left( \frac{1}{\Dmax}\right)^{1/16}~.
        \end{align*}
        The very last inequality is extremely lossy but sufficient to prove the lemma statement. 

        \medskip

        \noindent\textbf{\boldmath Case $|H(v)|< \DThres^{3/4}$:}      
        We obtain $|L(v)| \geq \DThres - \DThres^{3/4}$, to be used later in the proof. Let $L_v$ be the random variable that counts the successful children of $v$ in $L(v)$. 
        When considering only local minima in $N(v)$, we can assume that every node $u \in L(v)$ has degree $d_u = d$. This assumption can be made, since it only decreases the probability that there exists a local minimum in $N(v)$  and we use \Cref{lem:caseSmallDegreePreparation} of \Cref{lem:preparation} which only argues about the 1-hop neighborhood $N(v)$.

        With linearity of expectation and \Cref{obs:successful} we obtain $\mathbb{E}[L_v] = \frac{|L(v)|}{d}$.
        Since each child of $v$ is successful independently of all the other children, a Chernoff bound yields
        \begin{align}
             \mathbb{P}[L_v< \mathbb{E}[L_v]/2] < \exp(-\mathbb{E}[L_v]/8) < \exp(-|L(v)|/8d).
        \end{align}

        Now let $\ell = \mathbb{E}[L_v]/2 =  |L(v)|/2d$ 
        and apply \Cref{lem:caseSmallDegreePreparation} of \Cref{lem:preparation} to obtain:
        \begin{align*}
            \mathbb{P}[v \in U] & \leq \mathbb{P}[L_v<\ell] +\frac{1}{d\cdot \frac{|L(v)|}{2d}+1} \\
            & \leq \exp(-|L(v)|/8d) + \frac{2}{|L(v)|} \leq  \frac{3}{|L(v)|} \\
            & \leq \frac{3}{\DThres-\DThres^{3/4}} = \frac{3}{\Dmax^{3/4}-\Dmax^{9/16}} \stackrel{(*)}{\leq} \left( \frac{1}{\Dmax}\right)^{1/16}
        \end{align*}
        In $(*)$ we use that $\Dmax \geq 18$.
    \end{proof}

\begin{remark}\label{rem:lenzen}
    Similar to the proof of \Cref{lem:localSuccess} one can use \Cref{lem:preparation} to recover the central progress lemma of \cite[Lemma 4.1]{lenzen-tree-2011}.
    Their case corresponded to $|L(v)| \geq d_v/2 $ with threshold $d=d_v/16\log d_v$. In the proof of \Cref{lem:localSuccess} we bound the probability that $v \in U$ by $3/|L(v)|$. Plugging in the above numbers, we get an upper bound of $6/d_v$ compared to their claim of $5/d_v$.
\end{remark}
       \subsubsection{Shattering Behavior of $16c$ iterations of \LMJ}
       \label{sec:shattering}
    \Cref{lem:localSuccess} shows that the probability of a high-degree node $v$ to be removed is $1/\poly \Dmax$. As this is not a \emph{with high probability guarantee}, we use the shattering framework to prove that non-removed high-degree nodes in some sense form small connected components. A standard shattering proof would provide us with components of size $O(\poly\Dmax \log n)\gg \poly\log n$ if $\Dmax\gg \poly\log n$, which is insufficient for a $O(\log\log n)$ post-shattering phase. Inspired by \cite{CHLPU20},  we prove a slightly stronger result on the components reasoning about the existence of certain small dominating sets in each  component. Their small size is then exploited down the line in \cleanup in \Cref{sec:cleanup}.

     First we will show that well-spaced subsets of these nodes do not survive with a good probability. 
     
    \begin{lemma}
    \label{lem:survivingSubsets}
        Let $c > 0$ be an arbitrary constant. Let $W$ be the set of nodes defined in line 7 of procedure $\procdegreedrop(T,\DThres)$. For any 6-independent set $B$ the probability that $B\subseteq W$ is at most $\Dmax^{-c |B|}$. Here distance is measured with respect to $T$.
    \end{lemma}
    \begin{proof}  
Consider the $r=16c$ iterations of $\procdegreedrop(T,\DThres)$ and for $i=0,\ldots, r$ let $U_i$ be the set of nodes with degree higher than $\DThres$ that are not covered after the $i$-th iteration.  We have
\begin{align}
\mathbb{P}[B\subseteq W]=\mathbb{P}[B\subseteq U_r]= \prod_{i=1}^r \mathbb{P}[B\subseteq U_i \mid B\subseteq U_{i-1}]\stackrel{(*)}{\leq} \prod_{i=1}^r 1/{{\Dmax}^{|B|/16}}=1/{{\Dmax}^{|B|\cdot r/16}}\\
= 1/\Dmax^{c \cdot |B|} = \Dmax^{-c|B|}. 
\end{align}
At $(*)$ we use \Cref{lem:localSuccess} to bound $\Pr(B\subseteq U_i \mid B\subseteq U_{i-1})\leq (1/\Dmax^{1/16})^{|B|}$ for each $1\leq i\leq r$. 
Note that we can apply \Cref{lem:localSuccess} as the nodes in $B$ are 6-independent and thus have distance at least $7$, which implies that the probability that two different nodes of $B$ get covered in a single iteration of \LMJshort are independent. This is the case, since the probability that a node gets covered in one iteration of the for loop of $\procdegreedrop$ only depends on its 3-hop neighborhood: A node gets only covered if some node in its 2-hop neighborhood joins $S$ and wether a node joins $S$ depends on its 1-hop neighborhood.
    \end{proof}

The following lemma is standard, see e.g., \cite{FGLLL17}.\footnote{Still, we chose to re-prove the lemma as some versions of it in the literature have faulty proofs \cite{personalCommunication}  or too strong assumptions for our setting. For the proof of \Cref{lem:shatteringBase} see \Cref{sec:deferredproofs}} 

\begin{lemma}
    \label{lem:shatteringBase}
        Let $W$ be the set returned by $\procdegreedrop(T,\DThres)$. Let $c$ be the constant with which we run $\procdegreedrop$,i.e., we run the for loop $16c$ times. Let $c_1 > c-10$, then with probability at least $1-n^{-c_1}$ any $6$-independent $8$-connected set $W'\subseteq W$ has size at most $\log_{\Dmax} n$. The distances $6$ and $8$ are measured in $T$ and not in $T[W]$. 
    \end{lemma}
    
    \begin{lemma}
    \label{lem:dominatingSet}
        Let $W$ be the set returned by $\procdegreedrop(T,\DThres)$. Then with probability at least $1-n^{-c_1}$, for a tune-able constant $c_1$, any connected component of $T[W]$ has a 7-dominating set of size smaller than $\log_{\Dmax} n$. Where the distance is measured in $T$ and not $T[W]$.
    \end{lemma}
    \begin{proof}
        
        Let $C$ be any connected component of $T[W]$. We pick a 7-distance dominating set $D$ of $C$ greedily as follows:  We always pick an arbitrary node $v \in V(C)$ and remove all nodes in $N^T_6(v)\cap W$ from consideration for future picks. The resulting dominating set $D$ is by construction $6$-independent and $8$-connected (in $T$). 
        
        Thus by \Cref{lem:shatteringBase} the size of $D$ is with probability $1-n^{-c_1}$ at most $\log_{\Dmax} n$. 
     \end{proof}

\subsubsection{Proof of the Degree-Drop Lemma (Lemma~\ref{lem:degreedroptrees})}
\label{sec:lemDegreeDropProof}
    \begin{proof}[Proof of \Cref{lem:degreedroptrees}]
        We separately prove properties \ref{itm:1}--\ref{itm:4}.
\begin{compactitem}
\item \textbf{Property \ref{itm:1}:} No two adjacent nodes get added to $S$ in the same iteration of \LMJshort, since they cannot both be a local minimum at the same time. Furthermore, after a node gets added to $S$, all its neighbors get removed from $T$ and thus they cannot be added to $S$ in any later iteration. This shows that $S$ is an independent set.
 \item \textbf{Property \ref{itm:2}:} The set $W$ is computed as a subset of the remaining nodes after $S$ and $N_2(S)$ have been removed from the graph. Hence, no node of $W$ is adjacent to any node of $S$.

       \item \textbf{Property \ref{itm:3}:} Every node with degree higher than $\DThres$  gets added to $W$ in the last step of the algorithm. Hence, the maximum degree of $G[V(T) \setminus (S \cup N_2(S) \cup W)$ is at most $\DThres$.
        
       \item \textbf{Property \ref{itm:4}:} This is proven separately in \Cref{lem:dominatingSet}. \qedhere
       \end{compactitem}
    \end{proof}

    \subsection{Clean-up Phase (Post-shattering)} 
    \label{sec:cleanup}
The goal of this section is to show that we can run \Cref{alg:complete:cleanup} of \Cref{alg:complete} in $O(\log \log n)$ rounds.

A $(\lambda,\gamma)$-network decomposition is a decomposition of the nodes of a graph $G$ into $\lambda$ parts $V_1,\dots, V_\lambda$ such that each connected component of the induced graph $G[V_i]$ has diameter at most $\gamma$ for all $i$. 
    In particular if we have a network decomposition with $\gamma = O(\log \log n)$ we can compute an MIS in $O(\log \log n)$ 
    rounds. Process the color classes sequentially. Each color class consists of (possibly many) connected components of diameter $\gamma$ that are usually called clusters. Solve such a cluster by gathering its topology at a cluster leader and distributing the solution to all nodes of the cluster. In this section we build upon a similar post-shattering treatment as in \cite{CHLPU20}. The goal of this section is to compute a $(2,O(\log \log n))$-network decomposition for every connected component of $T[W_i]$  in $O(\log\log n)$ rounds yielding that we can compute the MIS $Z_i\subseteq W_i$ in $O(\log \log n)$ rounds.
   Due to \Cref{lem:parallelcleanup} (see below) we  can handle $W_1, \ldots, W_R$ in parallel.
    
    The following theorem, based on a rake \& compress procedure, from \cite{CHLPU20} tells us when we can compute a $(2,O(\log \log n))$-network decomposition efficiently on a tree that has a small $O(1)$-distance dominating set.

    \begin{theorem} \cite[Theorem~7]{CHLPU20}
    \label{thm:networkdecomp}
        Let $T$ be a tree for which there exists a $d$-distance dominating set of size $s$. There is a deterministic \local algorithm that computes a $(2,O(\log s+ d/k))$-network decomposition of $T^k$ in $O(k \log s + d + k \log^* n)$ rounds, i.e., $O(\log s + \log^*n)$ rounds when $d = O(1)$ and $k = O(1)$.
    \end{theorem}

     We show the following.

        \begin{lemma}
        \label{lem:networdDecompTree}
            Let $i \in \{1,\dots,R\}$ and let $W_i$ be the set $W$ returned by the $i$th call of \procdegreedrop($T,c$).
            There is a deterministic \local algorithm that computes a $(2,O(\log \log n))-$network decomposition for any connected component of $T[W_i]$ in $O(\log \log n + \log^* n)$ time.
        \end{lemma}
        \begin{proof}
            \Cref{lem:dominatingSet} tells us that every connected component of $T[W_i]$ admits a $7$-distance dominating set of size smaller than $\log_{\Dmax} n$, with the dominating distances measured in $T$.  Since $T$ is a tree, there exists one unique path between every pair of nodes of $W$. Thus,  any $7$-distance dominating set of a component of $T[W]$ with regard to distances in $T$ is also one with regard to distances in $T[W]$. 
            
            
            We use \cite[Theorem~7]{CHLPU20} with $d=7$, $k=1$ and $s= O(\log_{\Dmax}n)$ to obtain a $(2,O(\log \log n)$-network decomposition in $O(\log \log n + \log^*n) = O(\log \log n)$ rounds. 
        \end{proof}
    Next we show that different $W_i$s and $S_i$ do not interfere. 
        \begin{lemma}
            \label{lem:parallelcleanup}
                Let $1\leq i<j\leq R$ and let $W_i$ and $W_j$ be the sets $W$ returned in the $i$th and $j$th iteration of \Cref{alg:complete}, respectively. And let $S = \bigcup_{i=1}^R S_i$. Then the following hold
                \begin{compactenum}     
                    \item there is no edge between any node of $W_i$ and any node of $W_j$,\label{itm:independentWs}
                    \item there is no edge between any node of $W_i$ and any node of $S$.\label{itm:independentWandS}
                \end{compactenum}
        \end{lemma}
        \begin{proof}
            For \Cref{itm:independentWs}, observe that $W_i$ is always a subset of the current tree $T$. In particular, when $W_j$ gets defined $T$ does not contain $W_i \cup N(W_i)$, since we removed it in iteration $i$. Thus, there cannot be an edge between any node of $W_i$ and any node of $W_j$, $j>i$.
            
            For \Cref{itm:independentWandS}, we use \Cref{itm:2} of \Cref{lem:degreedroptrees} and thus know that there is no edge between any node $W_i$ and any node of $S_i$. For $j<i$ we know that $S_j \cup N_2(S_j)$ got removed before we compute $W_i$ and for $j>i$ we know that $W_i \cup N(W_i)$ got removed before we compute $S_j$. In both cases, we have one layer of nodes between $W_i$ and $S_j$, so in total, there is no edge between $W_i$ and $S$.
        \end{proof}

        \begin{lemma}
            \label{lem:cleanup}
            There is a deterministic \local algorithm that computes a MIS $Z$ of $T[W_1\cup \ldots \cup W_R]$ in $O(\log \log n)$ rounds, such that $Z\cup \bigcup_{i=1}^RS_i$ is independent.
            
        \end{lemma}
        \begin{proof}
        Due to \Cref{itm:independentWs} of \Cref{lem:parallelcleanup} the following can be executed in parallel for each $T[W_i]$, $i=1,\dots,R$. So we can fix one $W_i$ and show how we can compute a MIS for $T[W_i]$.
        
        In a similar fashion, We can deal with all connected components of $T[W_i$] in parallel, so we can fix one connected component $C$ and show how we obtain an MIS in $O(\log \log n)$ rounds.
        By \Cref{lem:networdDecompTree} we can compute a $(2,O(\log \log n))$-network decomposition of $C$ in $O(\log \log n)$ rounds. 
        Now we have two disjoint coloring classes $C_1$ and $C_2$, where each connected component, also called cluster, of $T[C_i]$ $(i=1,2)$ has diameter at most $O(\log \log n)$. We start with $C_1$ and gather, for each cluster of $T[C_1]$, the whole topology in one node, solve the cluster and then distribute the solution to every node of the cluster. This can be done in diameter,i.e., $O(\log \log n)$ rounds and gives us a MIS $S_1$ of $T[C_1]$. Now remove $S_1 \cup N(S_1)$ from $T$ and repeat the same process for the second coloring class $C_2$ and obtain a MIS $S_2$ of $T[C_2]$. Then. $S_1 \cup S_2$ is a MIS of $T[C]$.

        In total we get a MIS $Z$ of $T[W_1 \cup \dots \cup W_R]$ and by \Cref{itm:independentWandS} of \Cref{lem:parallelcleanup}, we know that there are no edges between any node of $Z$ and any node of $\bigcup_{i=1}^RS_i$ since we computed $Z$ on a set of nodes that already fulfills this property. Thus $Z \cup \bigcup_{i=1}^RS_i$ is independent.
        \end{proof}

     \subsection{Proof of \Cref{thm:rulingSetTrees}}

        Now we have gathered all the results we need to prove \Cref{thm:rulingSetTrees}.

        \begin{proof}[Proof of \Cref{thm:rulingSetTrees}]
            We show that \Cref{alg:complete} returns a 2-ruling set in $O(\log \log n)$ rounds with high probability.

            \textbf{Runtime.}  We first run $R=O(\log\log n)$ iterations of \procdegreedrop taking $O(1)$ rounds each due to \Cref{lem:degreedroptrees}. In \Cref{alg:complete:cleanup} we compute multiple MIS on all $G[W_1], \ldots, G[W_R]$ in parallel. It total this clean-up phase requires $O(\log\log n)$ rounds due to \Cref{lem:cleanup}. Computing an MIS on a constant-degree graph in \Cref{alg:complete:finalMIS} takes $O(\log^* n)$ rounds \cite{linial92}. 

\textbf{Termination and Correctness of Subroutine Calls.} By \Cref{lem:degreedroptrees}, the maximum degree of the remaining graph after iteration $i$ of \procdegreedrop is at most $\Delta^{(3/4)^i}$. Hence, after $O(\log\log \Delta)$ iterations the  maximum degree of the remaining graph is constant, allowing for the $O(\log^*n)$-round MIS algorithm in \Cref{alg:complete:finalMIS}. 

By \Cref{itm:independentWs} of \Cref{lem:parallelcleanup} all the $W_i$'s are non-adjacent and thus \Cref{alg:complete:cleanup} is equivalent to computing a MIS on $T[W_1,\dots,W_R$], which can be computed due to \Cref{lem:cleanup}.


\textbf{Independence.} The returned set is an independent set for the following reasons. \Cref{lem:degreedroptrees} guarantees that the sets $S_1,\ldots, S_R$ are independent sets and in \Cref{alg:complete:neighborremoval} we remove the $2$-hop neighborhood from the graph once some $S_i$ has been computed and hence no node in some $S_j$ with $j>i$ can have a neighbor in $S_i$.
By \Cref{itm:independentWandS} of \Cref{lem:parallelcleanup} every $W_i$ is non-adjacent to all the $S_j$ and thus the MIS computed in \Cref{alg:complete:cleanup} does not violate independence.
Also the MIS computed on the remaining nodes in \Cref{alg:complete:finalMIS} cannot violate independence.

\textbf{Domination.} All nodes contained in the ruling set and also for all nodes not included in the MIS computed in the clean-up phase in \Cref{alg:complete:cleanup} or the final MIS computation on the constant-degree graph \Cref{alg:complete:finalMIS} are contained in the returned set or have a neighbor in the set. The main part requiring a proof are the nodes ($S_i\cup W_i\cup N_2(S_i)\cup N(W_i)$) that are removed in \Cref{alg:complete:neighborremoval}. We have already reasoned about the nodes in $S_i$ and $W_i$ that are removed in this step. Nodes in $N_2(S_i)$ are dominated by the nodes in $S_i$. Each node $w\in N(W_i)$ has a neighbor $u\in W_i$. Now, either $u$ is in the MIS computed in the clean-up phase in \Cref{alg:complete:cleanup}, or $u$ has a neighbor in that MIS. In either case $w$ is dominated with distance at most $2$.

 \textbf{Error probability.}   All steps and properties of the algorithm actually hold deterministically (in particular the clean-up phase and the final MIS computation are fully deterministic), except for the shattering claim on the size of the $O(1)$-distance dominating set of each $W_1, \ldots, W_R$. This claim fails with probability $n^{-c_1}$ for a tuneable constant $c_1$, in which case the clean-up phase may still work, but does not run in $O(\log\log n)$ rounds. 
 \end{proof}

\section{Ruling sets in High Girth Graphs}

The goal of this section is to prove the following theorems and \Cref{cor:2rulingtreeV2}.
\thmRulingSetHighGirth*

\thmlogloglog*


    \begin{algorithm}[!htbp]
    \caption{Randomized 2-ruling set for high girth graphs}
    \label{alg:highgirth}
    \begin{algorithmic}[1]
       \State Initialize $S_i\leftarrow \emptyset$ for $i=1,\dots,R_1+1$. $P_i,Z_i,W_i \leftarrow \emptyset$ for $i=1,\dots R_2$.
       \For{$i=1, \dots, R_1=O(\log\log n)$}
       \State $S_i \leftarrow$ \procdegreedropsampling($G$, $\Dmax^{(3/4)^i}$)
       \State Remove  $S_i \cup N_2(S_i)$ from $G$ \label{alg:girth:neighborremoval1}
       \EndFor
       \For{$i=1, \dots, R_2=O(\log \log \Dsamp$)}\label{line:2ndforloop} // this is the exact same part as in \Cref{alg:complete} 
       \State $P_i,W_i \leftarrow \procdegreedropsampling(G,\Dsamp^{(3/4)^i})$
        \State Remove  $P_i \cup N_2(P_i) \cup W_i\cup N(W_i)$ from $G$ \label{alg:girth:neighborremoval2}
       \EndFor \label{line:end2ndfor}
         \State $S_{R_1+1} \leftarrow $MIS(G)  //using $O(\log^*n)$ rounds on remaining graph with constant maximum degree  \cite{linial92}\label{line:constantdegree}
       \State \textbf{for} $i=1,\dots,R_2$ in \textbf{parallel:} $Z_i \leftarrow $CleanUp($W_i$) \label{line:cleanupgirth}
       \State \textbf{return:} $S_{R_1+1}\cup \bigcup_{i=1}^{R_1} S_i \cup \bigcup_{i=1}^{R_2}(P_i \cup Z_i)$
    \end{algorithmic}
    
    \end{algorithm}
      \begin{algorithm}[!h]
    
      \begin{algorithmic}[1]
        \Procedure{\procdegreedropsampling}{$G,\DThres$}  \label{alg:degdrop:girth}
        \State $\emptyset \leftarrow S$ 
          \For{$j=1,\dots, \Tilde{c}$} //the constant $\Tilde{c}$ can be chosen such that \Cref{lem:degreedropgirth} becomes tunable
            \State $S \leftarrow$ $\LMJS(G) \cup S$
            \State Remove $S \cup N_2(S)$ from $G$
          \EndFor
          \State Return $S$
        \EndProcedure
      \end{algorithmic}
    \end{algorithm}

Intuitively we want to use the same approach like in \Cref{sec:trees}.
The main issue of our approach of \Cref{sec:trees} for high girth graphs is that we cannot expect the clean-up phase (\Cref{sec:cleanup}) to work, since we cannot exploit the existence of a $O(1)$-distance dominating set. Instead we adjust the procedure $\LMJ(T)$ by first sampling the set of active nodes. Further we add a second phase to each iteration of the adjusted procedure, where we chose a subset of active nodes deterministically. Using those two different phases we prove some similar results as in \Cref{sec:trees}, but instead of bounding the probability that a node does not get covered by $1/\poly \Delta$ we bound it by $<1/ \poly \ n$ for at least polylogarithmic $\Delta$ and thus obtain the degree drop immediately with high probability. We refer to the procedure as $\LMJS$, or  $\LMJshort$ in short.

Once degrees are at most polylogarithmic we revert to the method of \Cref{sec:trees}, but with a different clean-up phase. See \Cref{alg:highgirth} for pseudocode of the whole procedure.

\subsection{Degree Drop with Sampling for large degrees}
\label{sec:degreeDropHighGirth1}
The objective of this section is to prove \Cref{lem:degreedropgirth}. It shows that we can reduce the maximum degree from a starting high degree $\Dmax \geq \log^c n$ down to a lower maximum degree of $\Dsamp = O(\poly \log n)$ with high probability. The adjusted procedure \LMJS is presented in \Cref{alg:lmjs}.

\begin{lemma}\label{lem:degreedropgirth}
    Let $c>0$ be a sufficiently large constant. For a graph $G$ with girth 7 and maximum degree $\Dmax \geq \log^c n$ 
      there is a procedure that finds a set of nodes $S\subseteq V(G)$, such that with high probability the following holds
      \begin{compactenum}[(a)]
          \item S is an independent set,\label{item:indepgirth}
          \item the maximum degree $\DThres$ of $G[V(G)\setminus (S \cup N_2(S))]$ is at most $\Dmax^{3/4}$,\label{itm:degdropgirth}
      \end{compactenum}
      in $O(1)$ rounds in the LOCAL model.
  \end{lemma}

  The procedure of \Cref{lem:degreedropgirth} is denoted by $\procdegreedropsampling$.


 \begin{algorithm}[!h]
       \begin{algorithmic}[1]
        \Procedure{\LMJS}{$G$} \emph{// we use $\LMJSshort$ as short notation}
        \State $S \leftarrow \emptyset$
        \State \textbf{Phase 1:} \emph{//all nodes active}\label{alg:LMJSphase1}
        \State In parallel and independently for all $v\in V(G)$ mark a node as active with probability $1/2$
        \State Let $A \subseteq V(G)$ denote the set of marked nodes
        \State Uniformly and independently at random compute a real $r_v \in [0,1]$ for all $v\in A$
        \State In parallel for all $v \in A$ 
        \If{$r_v < r_w \forall w \in N(v) \cap A$}
        \State $S \leftarrow S \cup \{v\}$
        \EndIf
        \State \textbf{Phase 2:}   \emph{//only small degree nodes active}\label{alg:LMJSphase2}
        \State Set $A \leftarrow\{v \in V(G) \mid d_v \leq \DThres^{3/4} \} \setminus (S \cup N_2(S))$, i.e., every node with small degree that has not been covered marks itself as active
        \State Uniformly and independent at random compute a real $r_v \in [0,1]$ for all $v\in A$
        \State In parallel for all $v \in A$ 
        \If{$r_v < r_w \forall w \in N(v) \cap A$}
        \State $S \leftarrow S \cup \{v\}$
        \EndIf
        \State Return $S$
        \EndProcedure
      \end{algorithmic}
      \caption{}
      \label{alg:lmjs}
    \end{algorithm}

The following lemma is the core of the argument and shows that  \LMJSshort (Pseudocode in \Cref{alg:lmjs}) removes large-degree vertices with high probability.  
\begin{lemma}
\label{lem:girth:localsuccess}
For any $c_1>0$ and a sufficiently large constant $c>0$.
Let $\Dmax \geq \log^c n$ be the maximum degree of the input graph $G.$
    Consider a node $v \in V(G)$ at the beginning $\LMJSshort$ with $d_v \geq \Dmax^{3/4}$. And let $U$ be the set of nodes uncovered after one call of $\LMJSshort$. Then we have 
    \begin{align*}
        \mathbb{P}[v \in U] \leq n^{-c_1}~.
    \end{align*}    
\end{lemma}
    \begin{proof}
        We  split the proof in two cases: in the first one we  consider nodes with many high degree neighbors and argue about the 2-hop neighborhood of $v$ to show that there exists a node that joins the ruling set with high probability. In the second case we show  that if $v$ has many low degree neighbors, one of the direct neighbors of $v$  joins the ruling set with high probability.

        \begin{observation} \label{obs:independence} If a node $u$ is passive, the event that a children of $u$ joins the set $S$ is independent of that event for all other children of $u$.
        \end{observation}
        \textbf{1.Case:} Assume $v$ has at least $\Dmax^{0.6}$ neighbors with degree larger than $\sqrt{\Dmax}$.
        Here we want to argue that it is unlikely that $v$ does not to get covered due to \Cref{alg:LMJSphase1} of $\LMJSshort$.

        Let $H(v) \coloneqq \{ u\in N(v) \mid d_u > \sqrt{\Dmax }\}$. Observe that we can assume (for the analysis) that every $w \in H(v)$ actually has exactly degree $\sqrt{\Dmax}+1$.
        This assumption can be made, since we argue over the probability that a node $u$ with $dist(v,u)=2$ joins the ruling set and that assumptions only decreases the number of those nodes. Furthermore it does not change the probability of any of the nodes in $N_2^-(v)$ becoming a local minima.

Let $X$ be the random variable that counts the number of passive, i.e., non active, children $u \in H(v)$.
\begin{claim}\label{claim:passiveneighbors}
    With high probability $X\geq \Dmax^{0.6}/4$.
\end{claim}
\begin{proof}[Proof of \Cref{claim:passiveneighbors}]
Since each node is independently passive with probability $1/2$, We have
        \begin{align*}
            \mathbb{E}[X] = \frac{|H_v|}{2}
       \end{align*} 
        Since each node is independently active/passive we obtain by a Chernoff Bound
        \begin{align*}
            \mathbb{P}[X < \mathbb{E}[X]/2] < \exp\left(-\frac{\mathbb{E}[X]}{8}\right) =\exp\left(-\frac{|H_v|}{16}\right)\leq \exp\left(-\frac{\Dmax^{0.6}}{16}\right) \leq \exp(-\frac{\log^{0.6c}n}{16}) \leq n^{-c_2}~, 
       \end{align*}
       for a sufficiently large $c>0$.
       \renewcommand{\qed}{\ensuremath{\hfill\blacksquare}}
\end{proof}
\renewcommand{\qed}{\hfill \ensuremath{\Box}}

Let $P\subseteq H(v)$ be a set of children of $v$ that are passive with $|P|=\Dmax^{0.6}/4$. Formally, that is a random variable but such a set exists with high probability and we do not care which children of $H(v)$ are in $P$, since all nodes $u \in H(v)$ behave the same. This is the case, since they all have the same degree (by our assumption) and additionally we can also assume that all children of nodes in $H(v)$ have the maximum degree of $\Dmax$, because that only decreases the probability of any of them becoming a local minima. 

Let $Y$ be the random variable that counts the active children of the nodes in $P$. 
\begin{claim}\label{claim:active2hop}
With high probability $Y\geq\Dmax^{1.1}/16$. 
\end{claim}
\begin{proof}[Proof of \Cref{claim:active2hop}]
    Since each node $u \in P$ has $\sqrt{\Dmax}$ children and each of them is independently active with probability $1/2$, We have
     \begin{align*}
            \mathbb{E}[Y] = \frac{|P|\sqrt{\Dmax}}{2} = \frac{\Dmax^{0.6}\cdot\sqrt{\Dmax}}{8} = \frac{\Dmax^{1.1}}{8}
       \end{align*} 
        Since each node is independently active/passive we obtain by a Chernoff Bound
        \begin{align*}
            \mathbb{P}[Y< \mathbb{E}[Y]/2] < \exp\left(-\frac{\mathbb{E}[Y]}{8}\right) = \exp\left(-\frac{\Dmax^{1.1}}{256}\right) \leq \exp(-\frac{\log^{1.1c}n}{256}) \leq n^{-c_3} 
        \end{align*}
       for a sufficiently large choice $c>0$.
\renewcommand{\qed}{\ensuremath{\hfill\blacksquare}}
\end{proof}
\renewcommand{\qed}{\hfill \ensuremath{\Box}}
    Let $C\subseteq N_2^-(v)$ be a set of active children of children of $v$ that are passive with $|C| = \Dmax^{1.1}/16$. By \Cref{obs:successful} and \ref{obs:independence} all the nodes in $C$ join $S$ independently with probability $1/(\Dmax-1)$, since their degree is $\Dmax$, by our above assumption, and they have one passive neighbor.

    Thus, we obtain
    \begin{align*}
        \mathbb{P}[v \in U] \leq \left(1-\frac{1}{\Delta-1}\right)^{\Delta^{1.1}/16} \leq \exp\left( -\frac{\Delta^{1.1}}{16\Delta} \right) = \exp(-\Delta^{0.1}/16) =\exp(-\log^{0.1c}n/16)\leq n^{-c_4}
    \end{align*}

    For a sufficiently large $c > 0$.

    By a union bound over event that $X \geq \Dmax^{0.6}/4$, the event that $Y \geq \Dmax^{0.6}/16$ and the event that $\mathbb{P}[v \in U]$, we get \Cref{lem:girth:localsuccess} for this case.

        \textbf{2.Case:}
        Assume $v$ has at most $\Dmax^{0.6}$ neighbors with degree larger $\sqrt{\Dmax}$, and thus $v$ has at least $d \coloneqq \Dmax^{3/4}-\Dmax^{0.6}$ neighbors with degree at most $\sqrt{\Dmax}$.
        
        Let $L(v) = \{ w \in N(v) \mid d_w \leq \sqrt{\Dmax} \}$. We can assume that
        each of those nodes $w \in L(v)$ has degree exactly $\sqrt{\Dmax}$. 
        We can make this assumption, because it of only decreases the probability that a direct neighbor of $v$ joins the ruling set and we only argue about the probability that $v$ gets covered by a direct neighbor.
        
        In \Cref{alg:LMJSphase2} of $\LMJSshort(G)$, we activate all $w \in L(v) $ with probability 1 and know that $v$ itself is passive. By \Cref{obs:successful} and \ref{obs:independence}, we have $d$ nodes in $N(v)$ joining the ruling set with probability at least $1/(\sqrt{\Dmax}-1)$ independently.
        Thus we obtain,
        \begin{align*}
            \mathbb{P}[v \in U] & \leq \left(1-\frac{1}{\sqrt{\Dmax}-1}\right)^{d} \leq \left(1-\frac{1}{\sqrt{\Dmax}}\right)^{\Dmax^{3/4}-\Dmax^{0.6}}\\
            &
            \leq \left(1-\frac{1}{\sqrt{\Dmax}}\right)^{\Dmax^{0.55}} \leq \exp\left(-\frac{\Dmax^{0.55}}{\sqrt{\Dmax}}\right)\\
            &= \exp´\left(-\Dmax^{0.05}\right) =\exp\left(-\log^{0.05c}n\right) \leq n^{-c_1}~, 
        \end{align*}  
       for sufficiently large choice $c>0$. This completes the proof of \Cref{lem:girth:localsuccess}.
    \end{proof}

    As long as the maximum degree of the graph is large enough, \Cref{lem:girth:localsuccess} gives us the guarantee that high-degree nodes get covered with high probability. We are now ready to prove \Cref{lem:degreedropgirth}.

    \begin{proof}[Proof of \Cref{lem:degreedropgirth}]
    \textbf{Property \ref{item:indepgirth}:} 
    No two adjacent nodes get added to $S$ in the same iteration of $\LMJSshort$ since they cannot both be a local minimum at the same time. Furthermore, after a node gets added to $S$, all its neighbors are removed from $T$ and thus they cannot be added to $S$ in any later iteration. This shows that $S$ is an independent set.

    \textbf{Property \ref{itm:degdropgirth}:}
    Let $U$ be the set of uncovered nodes after \procdegreedropsampling. 
    By \Cref{lem:girth:localsuccess} a node $v$ with degree $d_v \geq \Dmax^{3/4}$ is not covered after one round of $\LMJSshort$ with probability at most $n^{-c_1}$. 
    In \procdegreedropsampling we run $\Tilde{c}$ iterations of $\LMJSshort$. 
    After one call of $\LMJSshort$ the degree of a node $v$ might decrease. If the degree drops below $\Dmax^{3/4}$ we do not need to consider it further. 
    If its degree is still higher then that \Cref{lem:girth:localsuccess} still holds for the next call of $\LMJSshort$.
    So if the degree of $v$ at the beginning of the last iteration of $\LMJSshort$ is still higher than $\Dmax^{3/4}$, we get that $\mathbb{P}[v \in U] \leq (n^{-c})^{\Tilde{c}}$.
    \end{proof}
    
\subsection{Degree drop for small degrees}
\label{sec:degreeDropHighGirth2}
Applying the procedure from \Cref{lem:degreedropgirth}, we reduce the maximum degree of the entire graph down to $\Dsamp=O(\poly \log n)$ in $O(\log\log\Delta)$ iterations.
    After the maximum degree got reduced to $O(\poly \log n)$ we only get a local success probability of $1/\poly \Delta$, see \Cref{sec:localSuccess}.
    Thus, similar to \Cref{sec:shattering}, we want to use the shattering framework to prove that non-removed high-degree nodes in some sense form small connected components. In contrast to \Cref{sec:cleanup}, we cannot exploit  a small $O(1)$-dominating set in all those components, instead we exploit that components have few nodes as, $\Dsamp$ is polylogarithmic. 

    \begin{lemma}
    \label{lem:girth:smallcomp}
        Given a high girth graph $G$ with max degree $\Dsamp = O(\poly \log n)$, there is a $O(\log \log \log n)$ round \LOCAL procedure that computes, with high probability, a subset $S \subseteq V(G)$, such that 
        \begin{compactenum}[(a)]
            \item $S$ is an independent set,\label{itm:girth:sampindepenent}
            \item each connected component of $G[V(G) \setminus (S\cup N_2(S))]$ has size at most $N=\Dsamp^6 \log n = O(\poly \log n)$.\label{itm:girth:smallcomp}
        \end{compactenum}        
    \end{lemma}

    The procedure of \Cref{lem:girth:smallcomp} is Line \ref{line:2ndforloop} to Line \ref{line:constantdegree} of \Cref{alg:highgirth}. 
     Since Line \ref{line:2ndforloop} to Line \ref{line:end2ndfor} of \Cref{alg:highgirth} are the same as Line \ref{line:treesforloop} to Line \ref{line:endtreeforloop} of \Cref{alg:complete}, and in all the proofs before the \Cref{sec:cleanup} we only need that we do not have cycles in the 3-hop neighborhood, the statements also hold for graphs of girth 7.
     Furthermore we use the same $O(\log^* n)$ round algorithm, in Line \ref{line:constantdegree}, for the remaining constant degree graph, since it works for general graphs, not just trees.   In particular a similar statement to \Cref{lem:shatteringBase}  holds.

    \begin{proof}[Proof of \Cref{lem:girth:smallcomp}]
    Define $S=S_{R_1+1} \cup \bigcup_{i=1}^{R_2}P_i $. We show that $S$ satisfies the properties stated in \Cref{lem:girth:smallcomp}.
    
    \textbf{Property \ref{itm:girth:sampindepenent}:}
            For all $i=1,\dots,R_2$ a node only joins a set $P_i$ if it is active and a local minimum among its active neighbors, therefore two adjacent nodes cannot join the same set $P_i$. Thus each $P_i$'s is independent. Then the 2-hop neighborhood gets removed from the tree (in $\LMJSshort$ and in \Cref{alg:girth:neighborremoval2} of \Cref{alg:highgirth}), thus there is no node between any node of $P_i$ and any node of $P_j$, for any $i \neq j $.
      Since we always remove the 2-hop neighborhod of each $P_i$, we can compute an MIS on the remaining graph in Line \ref{line:constantdegree} and the set $S$ stays independent.
        
    \textbf{Property \ref{itm:girth:smallcomp}:}
    The same proof as for \Cref{lem:shatteringBase} shows the following claim.

    \begin{claim}
    \label{claim:shatter}
    Let $W$ be the be the set of uncovered nodes after we ran line 6 to 10 of \Cref{alg:highgirth}. Let $\Tilde{c}$ be the constant with which we run \procdegreedropsampling. Then, we have with high probability that any 6-independent 8-connected set $W' \subseteq W$ has size at most $\log_{\Dsamp}n$. The distances 6 and 8 are measured in $G$ and not in $G[W]$.
 \end{claim}
    To complete the proof of \Cref{lem:girth:smallcomp}, assume there exists a connected component $C$ with $\Dsamp^6\cdot \log_{\Dsamp}n$ nodes. We now construct a $6$-distance dominating set violating \Cref{claim:shatter} as follows.
    Greedily pick a node $v \in V(C)$ and remove all nodes in $N_6^T(v)\cap W$ from consideration for future picks. The resulting dominating set $D$ is by construction $6$-independent and $8$-connected (in $G$). And in every step we remove at most $\Dsamp^6$ nodes from consideration, thus $D$ has size at least $(\Dsamp^6\cdot \log_{\Dsamp}n)/\Dsamp^6 = \log_{\Dsamp}n$.
    \end{proof}

\subsection{Proofs of \Cref{thm:rulingSetHighGirth} and \Cref{thm:logloglog}}
\label{sec:proofsThm2Thm3}
    \begin{proof}[Proof of \Cref{thm:rulingSetHighGirth} and \Cref{thm:logloglog}]
    We show that \Cref{alg:highgirth} returns a 2-ruling set in $\Tilde{O}(\log^{5/3} \log n)$ rounds with high probability, if we use the MIS algorithm presented in \cite{GG24} for the clean-up in Line \ref{line:cleanupgirth}, proving \Cref{thm:rulingSetHighGirth}.
    Further we show that \Cref{alg:highgirth} returns a $O(\log \log \log n)$-ruling set in $\Tilde{O}(\log \log n)$ rounds with high probability, if we use the $O(\log \log n)$ ruling set algorithm presented in \cite{GG24} for the clean-up in Line \ref{line:cleanupgirth}, proving \Cref{thm:logloglog}.
    
    \textbf{Runtime.} We first run $R_1=O(\log \log n)$ iterations of \procdegreedropsampling taking $O(1)$ rounds each due to \Cref{lem:degreedropgirth}. Then the remaining graph has maximum degree $\Dsamp = O(\poly \log n)$ and we run $R_2 = O(\log \log \Dsamp) = O(\log \log \log n)$ rounds of \procdegreedrop taking $O(1)$ rounds each due to \Cref{lem:degreedroptrees}. Computing an MIS on a constant degree graph in Line \ref{line:constantdegree} takes $O(\log^* n)$ rounds.
    Due to \Cref{itm:girth:smallcomp} we know that each connected component of each $G[W_i]$ has size at most $N=O(\poly \log n)$ and due to \Cref{lem:parallelcleanup} we can invoke any algorithm on each connected component of each $G[W_i]$ in parallel.

    For \Cref{thm:rulingSetHighGirth} run the MIS algorithm of \cite{GG24} which takes $\Tilde{O}(\log^{5/3} N) = \Tilde{O}(\log^{5/3} \log n)$.
     For \Cref{thm:logloglog} run the $O(\log \log N) = O(\log \log \log n)$-ruling set algorithm of \cite{GG24} which takes $\Tilde{O}(\log  N) = \Tilde{O}(\log\log n)$.

    \textbf{Termination and Correctness of Subroutine Calls.} By \Cref{lem:degreedropgirth}, the maximum degree of the remaining graph after iteration $i$ of \procdegreedropsampling is at most $\Dmax^{(3/4)^i}$. Hence, after $O(\log\log\Dmax)=O(\log\log n)$ iterations the maximum degree of the remaining graph is $\Dsamp = O(\poly \log n)$ w.h.p. By \Cref{lem:degreedroptrees}, the maximum degree of the remaining graph after iteration $i$ of \procdegreedrop is at most $\Dsamp^{(3/4)^i}$.
    Hence, after $O(\log \log \Dsamp)=O(\log \log \log n)$ iterations the maximum degree of the remaining graph is constant, allowing for the $O(\log^* n)$-round MIS algorithm in Line \ref{line:constantdegree}.

    By \Cref{itm:independentWs} of \Cref{lem:parallelcleanup} all the $W_i$'s are non-adjacent and thus in Line \ref{line:cleanupgirth} 
    we can compute a MIS, for \Cref{thm:rulingSetHighGirth}, or compute a $O(\log \log N)$ ruling set, for \Cref{thm:logloglog}, in parallel for each of the connected components of the $G[W_i]$. 
    For that we can use the respective algorithm from \cite{GG24}.

    \textbf{Independence.} The returned set is an independent set for the following reasons. \Cref{lem:degreedropgirth} guarantees that the sets $S_1, \dots, S_{R_1}$ are independent sets and \Cref{lem:degreedroptrees} guarantees that the sets $P_1,\dots,P_{R_2}$ are independent. In \Cref{alg:girth:neighborremoval1} and \Cref{alg:girth:neighborremoval2} we remove the 2-hop neighborhood from the graph once some $S_i$ or $P_i$ has been computed, respectively. Hence, no node in some $S_j$ or $P_j$ with $j>i$ can have a neighbor in $S_i$ or $P_j$, respectively. Furthermore each $P_i$ is non-adjacent to all the $S_j$, since the 2-hop neighborhoods of each $S_j$ gets removed before any $P_i$ gets computed. With the same reasoning, we obtain that every $W_i$ is non-adjacent to all the $S_j$. By \Cref{itm:independentWandS} of \Cref{lem:parallelcleanup} every $W_i$ is non-adjacent to all the $P_j$. And thus the MIS or $O(\log \log n)$ ruling set computed in Line \ref{line:cleanupgirth} does not violate independence. Also the MIS computed in Line \ref{line:constantdegree} cannot violate independence.

    \textbf{Domination:}
    All nodes contained in the ruling set, and also for all nodes not included in the MIS computed in Line \ref{line:constantdegree}, are contained in the returned set or have a neighbor in the set. With the nodes of the $W_i$'s we deal with in the clean-up phase in Line \ref{line:cleanupgirth}.
    We either compute an MIS or a $O(\log \log \log n)$ ruling set.
    
    If we compute an MIS those nodes are either contained in the returned set or have a neighbor in the set. 
    
    If we compute a  $O(\log \log \log n)$ ruling set, those nodes are either contained in the returned set or have a node in $O(\log \log \log n)$ distance.
    
    The main part requiring a proof are the nodes $(S_i\cup N_2(S_i))$ removed in \Cref{alg:girth:neighborremoval1} and the nodes $(P_i\cup W_i\cup N_2(P_i) \cup N(W_i))$ removed in \Cref{alg:girth:neighborremoval2}. We have already reasoned about the nodes in $S_i$, $P_i$ and $W_i$ that are removed in those steps. Nodes in $N_2(S_i)$ and $N_2(P_i)$ are dominated by the nodes in $S_i$ and $P_i$, respectively. Each node $w \in N(W_i)$ has a neighbor $u\in W_i$. 
    If we compute a MIS in the clean-up phase, $u$ either is in the MIS or has a neighbor in that MIS. In either case $w$ is dominated with distance at most 2.

    If we compute a $O(\log \log \log n)$-ruling set in the clean-Up phase, $u$ either is in that ruling set or is dominated by node with distance $O(\log \log n)$. In either case $w$ is dominated with distance at most $O(\log \log n +1) = O(\log \log n)$.

    \textbf{Error probability.} All steps and properties of the algorithm actually hold deterministically (in particular the clean-up phase, since both algorithms we use from \cite{GG24} are deterministically, and the final MIS computation are fully deterministic), except for the shattering claim on the size of the uncovered components and the degree drop of \procdegreedropsampling. These claims fail with probability $n^{-c}$ for a tunable constants $c_1$ and $c_2$.
    \end{proof}
    
\corruling*

\begin{proof}[Proof of \Cref{cor:2rulingtreeV2}]
    We first run  \Cref{alg:highgirth}, except for the clean-up phase (Line \ref{line:cleanupgirth}). For the first for loop we only need $R_1 = O(\log \log \log n) = O(\log \log \Dmax)$ iterations, to get the degree down to $O(\poly \log n)$.
    
     Our adaptation invokes the MIS algorithm for trees of \cite{Barenboim2010} as the clean-up phase in Line \ref{line:cleanupgirth}. By \Cref{lem:girth:smallcomp} we only consider components of size $N = O(\poly \log n)$. Then the MIS algorithm by \cite{Barenboim2010} works in $O(\log N/\log \log N) = O(\log \log n/\log \log \log n)$. Since we compute a MIS in the Clean-up phase, correctness of the whole algorithm follows from the proof of \Cref{thm:rulingSetHighGirth}. In total, we have a runtime of $O(\log \log \Dmax)+O(\log \log n/\log\log\log n)$.
\end{proof}

 \begin{remark}
        The \CONGEST model is identical to the \LOCAL model, except that the size of each message is restricted to  $O(\log n)$ bits. The pre-shattering phases of the algorithms for \Cref{thm:rulingSetTrees,thm:rulingSetHighGirth,cor:2rulingtreeV2} immediately work in \CONGEST. The same is true for the rake \& compress based post-shattering phase required to obtain the $O(\log\log n)$ round algorithm of \Cref{thm:rulingSetTrees}.
        Adapting the post-shattering phase in the result for  high girth graphs in \Cref{thm:rulingSetHighGirth} to the \CONGEST model requires more care. After the pre-shattering phase we are left with components of size $N = O(\poly\log n)$. We can compute an MIS of these components in $O(\poly\log N)=O(\poly\log\log n)$ by using one of the recent network decomposition algorithm, e.g., \cite{RG20,GGHIR23}, and using the MIS algorithm for small diameter graphs of \cite{CPS20} to solve the MIS problem efficiently in each cluster. Thus, in \CONGEST,  \Cref{thm:rulingSetHighGirth} holds with a slightly worse runtime of $\poly\log\log n$ rounds.
        The result of \Cref{cor:2rulingtreeV2} immediately holds in \CONGEST, as the post-shattering phase of \cite{Barenboim2010} can be implemented in \CONGEST. There is no natural way to extend \Cref{thm:logloglog}, since there is no similarly efficient deterministic $O(\log \log n)$-ruling set algorithm in the \CONGEST model.
   \end{remark}

\bibliographystyle{alpha}
\bibliography{references}

\appendix
\allowdisplaybreaks
\section{Proof of Lemma~\ref{lem:preparation} from Section~\ref{sec:trees}}
\label{app:probability}


\paragraph{The distribution of the Minimum of independent Random Variables.} The following lemmata are crucial for our analysis, to capture the influence of dependencies.  
    
    \begin{lemma}[CDF of Minimum of RVs, e.g., \cite{StackExchangeMinDistribution}]
    \label{lem:minimumDistribution}
Let $X_1, \ldots, X_k$ be i.i.d. random variables where the CDF of each $X_i$ is $F: D\rightarrow [0,1]$. Then $\min_{i=1}^k X_i$ has CDF $1-(1-F(x))^k$ where $x\in D$. 
    \end{lemma}

\begin{lemma} 
\label{lem:conditionalProb}The following statements hold:
\begin{enumerate}
\item Let $r$ and $X_1,\ldots,X_k$ be uniformly distributed in $[0,1]$. 
Let $M=\min_{i=1}^k X_i$. Then, the density function of $r$ conditioned on $r<M$ is 
$f_{r\mid (r\leq M)}(x)=(k+1)(1-x)^k$.\label{itm:cond1}
\item    Let $r$ be uniformly distributed in [0,1] and let $X_1, \dots, X_\ell$ be distributed according to the probability density function $f(x)=k(1-x)^{k-1}$. Let $M = \min_{i=1}^\ell X_i$. Then, $\mathbb{P}[r<M] = \frac{1}{(k \ell+1)}$.\label{itm:cond2}
\end{enumerate}
\end{lemma}

In \Cref{lem:conditionalProb} (part \ref{itm:cond1})  $X_1,\dots,X_k$ can be seen as the random variables $r_w$ for children $w$ of $v$ and the conditioned density function tells us the probability that none of the children of $v$ is a local minimum, due to the fact that $r_v$ is smaller than all the values $r_w$. \Cref{lem:conditionalProb} (part \ref{itm:cond2}) gives us a similar statement for successful nodes in $N_2^-(v)=\{u\in V\mid dist(v,u)=2\}$.

\begin{proof}[Proof of \Cref{lem:conditionalProb}]
\textbf{Proof of \Cref{itm:cond1}:}

As the pdf of each $X_i$ is $f_{X_i}(x)=f(x)=1$, we obtain by \Cref{lem:minimumDistribution} that $f_M(x)= k\cdot (1-F(x))^{k-1}\cdot f(x)=k\cdot (1-x)^{k-1}$ holds. 
We also have
\begin{align}
\label{eqn:a2}
\mathbb{P}[r=x \wedge r<M]=\int_x^1 f_{r,M}(x,m) dm = \int_x^1 f_r(x) f_M(m)dm = \int_x^1 1 \cdot k (1-m)^{k-1} dm = (1-x)^k
\end{align}
We deduce
\begin{align*}
\mathbb{P}[r<M]= \int_{0}^1 \mathbb{P}[r=x \wedge r<M] dx \stackrel{(\ref{eqn:a2})}= \int_{0}^1 (1-x)^k dx=1/(k+1)~.
\end{align*}
Thus, we obtain
\begin{align*}
f_{r\mid (r\leq M)}(x)=\frac{\mathbb{P}[r=x \wedge r<M]}{\mathbb{P}[r<M]}=\frac{(1-x)^k}{1/(k+1)}=(k+1)(1-x)^k~.
\end{align*}
\textbf{Proof of \Cref{itm:cond2}:}

    As the CDF of each $X_i$ is 
    \begin{align*}
        F(x)=\int_0^xf(z)dz = \int_0^xk\cdot (1-z)^{k-1}dz = -(1-z)^k|_0^x =1 -(1-x)^k,
    \end{align*}  
    we have, by \Cref{lem:minimumDistribution}, for the CDF $F_M(x)$ of M:
    \begin{align*}
        F_M(x)=1-(1-F(x))^\ell = 1-(1-(1 -(1-x)^k))^\ell= 1-(1-x)^{k\ell},
    \end{align*}
    and thus $f_M(x)=k\ell (1-x)^{k\ell -1}$.
    We also have
    \begin{align*}
        \mathbb{P}[r=x \land r<M] & =\int_x^1f_{r,M}(x,m)dm=\int_x^1f_r(x)f_M(x)dm\\
        & =\int_x^1 1\cdot k\ell(1-m)^{k\ell - 1}=-(1-m)^{k\ell}|_x^1=(1-x)^{k\ell}
    \end{align*}
    We deduce
    \begin{align*}
        \mathbb{P}[r<M]=\int_0^1\mathbb{P}[r=x \land r<M]dx=\int_0^1(1-x)^{k\ell}=-\frac{1}{k\ell+1}(1-x)^{k\ell+1}|_0^1=\frac{1}{k\ell+1}.
    \end{align*}
    \end{proof}

\lemPreparation*
    \begin{proof}
    We say a set $W\subseteq V(T)$ is successful, if $\forall w \in W:w$ is successful. 


For a parameter $x$ let $f^x$ be the density function of $r\in [0,1]$ chosen u.a.r conditioned on the event that $r$ is smaller than the minimum of $x$ values chosen u.a.r. from $[0,1]$. Let $\ell$ be an integer and let $X_1,\ldots,X_{\ell}$ be distributed according to $f^{d-1}$ and $Y_1,\ldots, Y_{\ell}$ according to $f^{\Dmax-1}$. In order to show the lemma we will show the following two statements.

  \begin{align}
            \mathbb{P}[v \in U] \leq \mathbb{P}[L_v<\ell]+\mathbb{P}[r_v < \min_{i=1}^\ell X_i ]\leq \mathbb{P}[L_v<\ell] +\frac{1}{d\cdot \ell+1}\label{lem:caseSmallDegreePreparationProof}
        \end{align}

  \begin{align}
  \mathbb{P}[v \in U] \leq \prod_{u\in H(v)}\left(\mathbb{P}[C_u< \ell] + \mathbb{P}[r_u < \min_{i=1}^\ell Y_i]\right) \leq 
            \prod_{u\in H(v)}\left(\mathbb{P}[C_u<\ell] +\frac{1}{\Dmax\cdot \ell+1}\right)
          \label{lem:caseHighDegreePreparationProof}
        \end{align}

        \textbf{Proof of \Cref{lem:caseSmallDegreePreparationProof}}: We know that $v$ is covered if it has a local minimum in $N(v)$. Further more we will, for the proof of \Cref{lem:caseSmallDegreePreparationProof}, only bound the probability that $v$ has no child that is a local minimum.
        Therefore we can make the following assumption, since it only decreases the probability that $v$ has a child that is a local minimum:  
        Assume that every node $u\in L(v)$ has degree exactly $d$.
        
        There is a local minimum in $N(v)$ if there exists a successful node $w\in N(v)$ with $r_v > r_w$ (since then $w$ is a local minimum). Thus, in order for $v$ to be not covered, for every successful neighbor $w$ it has to hold that $r_v < r_w$. We obtain with the law of total probability. 
        \begin{align*}
            \mathbb{P}[v \in U ] 
            &\leq \mathbb{P}[\text{$v$ has no child that is a local minimum}] \\
            &=  \mathbb{P}[\text{$v$ has no child that is a local minimum} \mid L_v<\ell]\cdot \mathbb{P}[L_v < \ell]\\ 
            &+ \mathbb{P}[\text{$v$ has no child that is a local minimum} \mid L_v\geq \ell]\cdot \mathbb{P}[L_v \geq \ell]\\
            &\leq \mathbb{P}[L_v < \ell]
            + \sum_{W\subseteq N(v):|W| \geq  \ell}\mathbb{P}[r_v < \min_{w \in W} r_w\mid  W \text{ successful}] \cdot \mathbb{P}[W \text{ successful}]  \\
            &\stackrel{(*)}{\leq} \mathbb{P}[L_v < \ell] + \sum_{W\subseteq N(v):|W| \geq  \ell}\mathbb{P}[r_v < \min_{i=1, \dots, \ell} X_i  ] \cdot \mathbb{P}[W \text{ successful}]\\
            & \leq  \mathbb{P}[L_v<\ell] + \mathbb{P}[r_v < \min_{i=1, \dots, \ell} X_i  ]\cdot \sum_{W\subseteq N(v):|W| \geq  \ell} \mathbb{P}[W \text{ successful}]  \\ 
          &  \leq \mathbb{P}[L_v<\ell]+\mathbb{P}[r_v < \min_{i=1, \dots, \ell} X_i  ] \stackrel{\Cref{lem:conditionalProb}}{\leq} \mathbb{P}[L_v < \ell] + \frac{1}{d\cdot \ell+1}.
        \end{align*}
        In $(*)$ we use the fact, that conditioning on the success of $W$ only influences the random variables $r_w$ with $w \in W$ and not the distribution of $r_v$. Also the distribution of $r_w$ for each $w\in W$ remains independent even if we condition on the set $W$ being successful.
        
        \textbf{Proof of \Cref{lem:caseHighDegreePreparationProof}}: We know that $v$ is covered if there is a local minimum in $N_2^-(v) \subseteq N_2(v)$. Further more we will, for the proof of \Cref{lem:caseHighDegreePreparationProof}, only bound the probability that there is no $w \in N_2^-(v)$ that is a local minimum. Therefor we can make the following assumptions, since it only decreases the probability that there is a node $w\in N_2^-(v)$ that is a local minimum:  Assume that every node $u\in H(v)$ has degree $d+1$, this decreases the size of $N_2^-(v)$ and assume that every node $w \in N_2^-(v)$ has degree $\Dmax$.
        
        There exists a successful node $w\in N_2^-(v)$ with $r_{u_w}>r_w$,  where $u_w$ is the unique child of $v$ that is the parent of $w$. Thus, in order for $v$ to be not covered, for every successful node $w \in N_2^-(v)$ it has to hold that $r_{u_w}<r_w$. Now let us first look at one fixed node $u \in H(v)$ with degree $d_u$ at least $d$ 
        Then we obtain by \Cref{lem:caseSmallDegreePreparationProof}:
        \begin{align*}
        \label{eqn:uinH}
            \mathbb{P}[\text{$u$ has no child that is local minimum}] \leq \mathbb{P}[C_u< \ell] + \mathbb{P}[r_u < \min_{i=1,\dots,\ell} Z_i].
        \end{align*}
        Note that the events whether different nodes in $H(v)$ have children that are local minima are independent. Hence, we obtain. 
:
        \begin{align*}
            \mathbb{P}[v \in U] 
           & \leq \prod_{u \in H(v)} \mathbb{P}[u \text{ has no child that is local a minimum}] \\
           & \leq\prod_{u\in H(v)}\left(\mathbb{P}[C_u< \ell] + \mathbb{P}[r_u < \min_{i=1}^\ell Y_i]\right) \\
            & \stackrel{(*)}{\leq} \prod_{u\in H(v)}\left(\mathbb{P}[C_u<\ell] +\frac{1}{\Dmax\cdot \ell+1}\right)~, 
        \end{align*}
    where we used \Cref{itm:cond1} and \Cref{itm:cond2} of \Cref{lem:conditionalProb} at $(*)$.
    \end{proof}
\section{Deferred Proofs}
\label{sec:deferredproofs}
\begin{proof}[Proof of \Cref{lem:shatteringBase}]
        Let $H \coloneqq T^{[7,8]}$ be the graph on $V(T)$ with the edge set $\{ \{u,v\} \mid 7 \leq dist_T(u,v) \leq 8\}.$ 
        
        Assume that there exists a $6$-independent $8$-connected set $W'$ of size bigger than $\log_{\Dmax} n$.
        Then we know that $H[W]$ contains a tree on more than $\log_{\Dmax}n$ nodes. There are at most $4^{\log_{\Dmax}n}$ different such tree topologies and each can be embedded into $H$ in less than $n \cdot \Dmax^{8(\log_{\Dmax}n-1)}$ ways:
        there are $n$ choices for the root and at most $\Dmax^{8}$ choices for each further node.
        
        Since we measure the distances in $T$ and not in $H$, by \Cref{lem:survivingSubsets}, the probability that a particular tree occurs in $H[W]$ is at most $\Dmax^{-c  \log_{\Dmax}n}$.
        
        A union bound over all trees thus lets us conclude that such a set exists with probability at most $4^{\log_{\Dmax}n} \cdot n \cdot \Dmax^{8(\log_{\Dmax}n-1)} \cdot \Dmax^{-c \log_{\Dmax}n} \leq n^{-(c-10)} \leq n^{-c_1}$. Since $c$ in \Cref{lem:survivingSubsets} is tunable $c_1$ in this statement is as well. Thus with probability $1-n^{-c_1}$ such a set does not exist. 
    \end{proof}

\section{Concentration Inequalities and Density Functions}

For a proof of the Chernoff Bound, see e.g. \cite{DGP98}
\begin{lemma}[Chernoff Bound]
    Let $X$ be the sum of n independent, identically distributed indicator random variables. For any $\delta \in [0,1]$,
    \begin{align*}
        \mathbb{P}[X<(1-\delta)\mathbb{E}[X]]<\exp(-\delta^2 \mathbb{E}[X]/2).
    \end{align*}
    
\end{lemma}

\begin{definition}
    The cumulative density function (CDF) of a random variable $X$ is the function given by
    \begin{align*}
        F_X(x)=\mathbb{P}[X\leq x].
    \end{align*}
    The (probability) density function (pdf) of a random variable $X$ is the function given by
    \begin{align*}
        \mathbb{P}[a\leq X\leq b] = \int_a^bf_X(x)dx.
    \end{align*}
    
\end{definition}

\end{document}

    \section{Idea: High girth sampling}
    {\color{blue}
Consider the following algorithm with 2 steps from the perspective of node $v$.
\begin{itemize}
    \item \textbf{Large 2-hop:} Each node marks itself as active independently with probability $1/2$. Then, each active node picks a random number between 0 and 1 and each local minima is chosen into an independent set. Kill each node with a neighboring independent node.
    \item \textbf{Small degree neighbors:} Consider each node with degree at most $\sqrt{\Delta}$ as active. Each active node picks a random real between $0$ and $1$ and the local minima are chosen into an independent set.
\end{itemize}

\textbf{Goal:} Each node with degree at least $\Delta^{3/4}$ dies.
We split into two cases. First, consider nodes with a large amount of two-hop neighbors. Suppose that $\Delta$ is at least some large $\poly \log n$.

\noindent \textbf{Many 2-hop nodes:} Suppose that there are at least $\Delta^{0.6}$ neighbors with degree at least $\sqrt{\Delta}$.
Using Chernoff and the Union bound, we get that each of the direct neighbors has at least $1/3$ neighbors are active.
Using Chernoff again, we have that at most half of the direct neighbors of $v$ with degree at least $\sqrt{\Delta}$ are active.
Combining the above, we have that there are at least $\Delta^{0.55}$ passive neighbors with at least $\sqrt{\Delta}/3$ active neighbors with high probability.
The actions of the active neighbors of passive nodes are independent, since they do not share any neighbors (passive nodes are not participating!). Now, with very high probability, at least one $2$-hop neighbor of $v$ is a local minimum.

\noindent \textbf{Small degree neighbors:} 
Suppose that a node $v$ survived the first step.
Furthermore, assume that no neighbor was removed in the first step. Otherwise, we would have a minimum in our 2-hop neighborhood.
Now, it must be the case (whp) that $v$ has at least $\Delta^{0.6}$ neighbors of degree at most $\sqrt{\Delta}$. Otherwise, $v$ does not survive the first step.

By design, $v$ is not active and will not participate in finding minima. Hence, the trials of the neighbors are independent. We have $\Delta^{0.6}$ independent trials with error probability at most $1 - 1/\sqrt{\Delta}$. Chernoff says that $v$ is killed with high probability.

}

\end{document}

    \section{Improved MIS in high girth graphs}
The goal of this section is to prove the following theorem. \\
\mb{can we compute $(\alpha,\beta)$-ruling set from a $(1,\beta)$-ruling set? }
\thmMISHighGirth*   

\section{Outdated/Wrong/Ideas}

    \begin{lemma}
        Let $c > 0$ be an arbitrary constant. At the end of phase $i$ of step 2, with probability at least $1-n^{-c}$ all nodes with degree higher than $\Delta_i$ are covered and the maximal degree of the active nodes is at most $\Delta_i$ after phase $i$ of step 2 with probability $1-n^{-c}$.
    \end{lemma}
 
    \begin{proof}
        We condition on the event, that the previous phase was successful, i.e. there are no nodes with degree higher than $\Delta_{i-1}$, so $\Delta_{i-1}$ is our maximum degree in phase $i$. For $i=1$ this happens with probability 1 and for phase $i$ we will inductively see that this happens with probability $1-n^{-c}$.\\  
        Let $\tilde{c} = c+1$ and let $k_i = 2\tilde{c}$.\\
        Consider an active node $v \in W_i = \{ v \in V \mid d_v > \Delta_i \} \cap A$, in particular $d_v \in (\Delta_i,\Delta_{i-1}]$.\\
        \textbf{1.Case}: Assume that $\sum_{w \in N(v)} (d_v-1) \geq \Delta_{i}^{1.5}$.\\ 
        In this case we obtain that $|N_2^-(v)| \geq \Delta_i^{1.5}$ (nodes $w$ with $dist(v,w) =2 $), since our input graph is a tree(has girth at least 5). Thus we obtain:
        \begin{align*}
            \mathbb{P}[\text{$v$ remains uncovered after phase $i$}] = \mathbb{P}[\text{$v$ is not covered after $k_i$ rounds of LMJ]} \leq (\prod_{w \in N_{2}(v)} (1-\frac{1}{1+d_w}))^{k_i} \\
            \leq (\prod_{w \in N_{2}(v)} (1-\frac{1} {1+d_w}))^{k_i} \leq (1-\frac{1}{1+d_w})^{\Delta_i^{1.5} \cdot k_i} \leq \exp(-\frac{\Delta_i^{1.5}k_i}{2d_w}) 
            \leq \exp(-\frac{\Delta_i^{1.5}k_i}{2\Delta_{i-1}})\\
            = \exp(-\frac{\Delta_{i-1}^{1.5 \cdot 3/4}k_i}{2\Delta_{i-1}}) 
            = \exp(-\Delta_{i-1}^{\frac{1}{8}}{\tilde{c}}) \leq \exp(- \tilde{c} \log n) = n^{-\tilde{c}} .
        \end{align*}
        \textbf{2.Case}: Assume that $\sum_{w \in N(v)}(d_w-1) < \Delta_i^{1.5}.$\\
        Thus $v$ has at most $\Delta_i^{3/4}$ neighbors with degree at least $\Delta_i^{3/4}$ and since $d_v > \Delta_i$ it has at least $\Delta_i-\Delta_i^{3/4}$ neighbors with degree at most $\Delta^{3/4}$. Let $A \coloneqq \{ w \in N(v) \mid d_w \leq \Delta^{3/4}\}$, then:

        \begin{align*}
            \mathbb{P}[\text{$v$ remains uncovered after phase $i$}]=\mathbb{P}[\text{$v$ is not covered after $k_i$ rounds of LMJ]} \\
            \leq (\prod_{w \in N_{2}(v)} (1-\frac{1}{1+d_w}))^{k_i} 
            \leq (\prod_{w \in A}(1-\frac{1}{d_w+1})^{k_i} 
            \leq (1-\frac{1}{\Delta^{3/4}+1})^{(\Delta_i - \Delta_i^{3/4})k_i}
            \leq \exp(-\frac{(\Delta_i-\Delta_i^{3/4})k_i}{2\Delta_i^{3/4}}) \\
            \leq \exp(-(\Delta_i^{0.25}-1)\tilde{c}) 
            = \exp((-\Delta_{i-1}^{\frac{3}{16}}-1)\tilde{c}) 
            \leq \exp(-\Delta_{i-1}^{\frac{1}{8}}\tilde{c}) \leq \exp(-\tilde{c} \log n) = n^{-\tilde{c}}.
        \end{align*}
        Thus with probability $1-n^{-\tilde{c}}$ every node in $W_i$ is covered after phase $i$ of step 2, which implies that $W_i \subseteq N_2(S_i)$ for $S_i$ being the ruling set $S$ at the end of phase $i$ of step 2. Therefor the set of active nodes $A$ does not contain any node of $W_i$ with high probability, so all nodes remaining in $A$ have at most degree $\Delta_i$ with probability $1-n^{-\tilde{c}}$.\\
        We now drop the conditioning on the event, that phase $i$ was successful and use the union bound to obtain the lemma.
    \end{proof}

    SHOULD WORK AGAIN:

    Finish-off: After $R$ phases of step 2 we know that the maximum degree in $G[A]$ is, with probability $1-n^{-c}$, at most $\Delta_R = \Delta^{(\frac{3}{4})^{R}} \leq \log^8 n$.\\
    Again observe, that we can just compute a 2 ruling set on the remaining nodes, without running at risk to violate the independency of our solution, since we always removed the 2-hop neighborhood of every intermediate solution.\\
    So our goal is to compute a 2 ruling set on the remaining active nodes:
    For that we use the work of Ghaffari, which shows, that we can compute a 2 ruling set in $O(2\log^{\frac{1}{2}} \Delta)+O(2^{\sqrt{\log \log n}}\textit{})$. If we were to plug in our max degree $\Delta_R$, we would obtain a running time of $O(2\log^{\frac{1}{2}} \Delta_R)+O(2^{\sqrt{\log \log n}})$ = $O(\sqrt{\log \log n})+O(2^{\sqrt{\log \log n}})$.
    The critical part is the second summand of the runtime, but we can get rid of this exponential term, because we just consider trees (bounded arb./bounded girth?).\\
    In their presented algorithm they use a sparsification routine (super fast t ruling set paper) and compute an MIS on the remaining graph (faster MIS is their contribution), showing the following result:

    BEFORE WE APPLY THE MIS ALGO WE DO ONE STEP OF SPARSIFICATION (THAT IS THE 2 RULING SET APPROACH). I.e the degree will be further decreased to $O(f_\$)$
    
    \begin{lemma}
        Let c be a large enough constant and B be the set of nodes remaining undecided after $\Theta(c \log \Delta)$ rounds of the MIS algorithm. Then, with probability at least $1-n^{-c}$, we have the following two properties:
        \begin{itemize}
            \item[(P1)] There is no $(G^{4^-})$-independent $(G^{9^-})$-connected subset $S \subseteq B$ s.t. $|S| \geq \log_\Delta n$.
            \item[(P2)] All connected components of $G[B]$, have each at most $O(\log_\Delta n \cdot \Delta^4)$ nodes.
        \end{itemize}
    \end{lemma}

    For our use-case we can plug in $\Delta = \Delta_R $ and obtain that each connected component of $G[B]$ has size at most $O(\log_{\Delta_R}n \cdot \Delta_R^4)=O(\log_{\Delta_R}n \cdot \log^{32} n)=O(\frac{\log^{33}n}{\log \log^8 n}) =O(poly \log n) $.
    So let $N=\frac{\log^{33}n}{\log \log^8 n}$.\\
    
    Since we are considering trees, we know that the arboricity $a(G) = 1$, thus we can compute a H-partition of size $O(\log N)$ and degree at most $(2+ \epsilon) \cdot 2$ in $O(\log N)$ time ($0 < \epsilon \leq 1$). Now we can compute a valid coloring $W_i$ for each $H_i$ with at most $25 \geq ((2+ \epsilon) \cdot 2)^{2}$ colors ($ i \in [ \ell ]$, for $\ell = O(\log N)$).
    Now we can compute an MIS using this coloring as follows:
    Initialize $S \coloneq W_{\ell 1}$ (first color-class of $W_{\ell}$) and extend it to a valid MIS for $H_{\ell}$, by iterating through all color-classes and checking each node $v$ (in parallel) of that color-class weither $v$ can join $S$. \\
    Now we can extend the MIS to an MIS of the component by iterating through all layers $H_{\ell -1}, \dots , H_1$ and repeating the above procedure (i.e. going through all color-classes).\\
    This procedure takes $O(\log N)$ time, since we have $O(\log N)$ layers and each is done in $25 = O(1)$ rounds.\\
    Thus in total we get a total runtime of $O(\log N + \log N) = O(\log poly \log n) = O(\log \log n)$, for computing the H-partition and the MIS.\\

\section{Observations}
    General assumption: input $T$ is a tree.
\begin{itemize} 
    \item To Do:
    \begin{itemize}
        \item \textbf{Fix dependencies!}, look at star example. $\rightarrow$ second case (large 2-hop neighborhood)! 
        \item  obtain shattering for each round (for nodes with high degree)
    \end{itemize}

    \item Can solve shattering instances parallel since they are independent, because we remove 2-hop neighborhood.

    \item If $\Delta \leq O(\log \log n)$, we are done because of $O(\Delta + \log^* n).$
    \item If we are given a subrgaph $H \subset G$ and we compute a $t$-ruling set on $H$. Then we have a a $t+x$-ruling set for all nodes in $G$ with distance at most $x$ to $H$.\\
    In particular a MIS on a subgraph $H$ gives us a $\beta$-ruling set for all $v \in V(G)$ with $dist(v,H)\leq \beta -1$. \\
    (Usecase: compute MIS on high degree nodes?)
    \item If $T$ is $\Delta$-regular, we have an $O(\log \log n)$ rounds algorithm ($\beta \geq 3$). In particular $O(1)$ for $\Delta \geq \log^{\frac{1}{\beta-2}} n$ and for the other case we need shattering. (local minima join)
    \item Still needs formal proof: If all degrees lie in the interval $[\Delta^{\frac{1}{D}},\Delta]$ and $\beta \geq D+2$, we get an $O(\log \log n)$ algorithm with same case distinction as above. In particular $O(1)$ algo. for $\Delta \geq \log^\frac{1}{\beta-2} n$. \\
    (Need to work out the details in order of constants in runtime, in particular what do we get for arbitrary graphs, i.e. $D= \Omega(log \Delta)$? It should be $O(poly(\beta) \log \log n)$ and since $\beta \geq D + 2 \geq \log \Delta +2$, we dont get a $O(\log \log n)$ algo anymore)\\
    (local minima join)
    \item If all degrees are in $[x,x^{1.1})$, we get shattering wrt to $x$, i.e. prob. that a node after $O(1)$ rounds is not covered is at most $x^{-1}$.
    \item Idea to subdivide the nodes into degree classes, e.g. $[x,x^{1.1})$ for $x = \Delta, \dots , 1$ or $x=1, \dots , \Delta$ brings the following problem: it is highly unclear, what the degree of the neighbors of a node $v$ with $d(v) \in [x,x^{1.1})$ is. Also if you look at the induced subgraph on nodes in that degree range, the degree of each node can get arbitrarily small.
    \item To Do: Look at the counterexample: are there activation probabilities $p_v$ for all $v \in V$ s.t. one layer or on node of a layer shatters (maybe just dependent of a constant amount of other layers). $\rightarrow $ has at most diameter of $O(\log \log n)$
    \item Revisist sampling/sparsification probabilities again (can fast 3 ruling sets approach be generalised?)\\
    \textbf{Problem:} if we try to use same sparsification as Kothapalli and Pemmaraju (super-fast 3-ruling sets) we have to sparsify faster (in $O(\log \log n)$ rounds, which leads to a higher degree in the components. In particular I could only bound it by $2^{i}-2^{i-1}$ in stage $i$ which can get "arbitrary" big ($1 \leq i \leq O(\frac{\log \log \Delta}{(\log n)^\epsilon}) \leq O(\log \log n)$).\\
    I am not completely sure why we need the $\epsilon$ in the original paper $\rightarrow$ understand that.
    \item How do sparsification and activation probabilities differ? $\rightarrow$ with activation prob. nodes can change there status each round, where in sparsification each node is permanently added to one nodeset which will be solved later.
    \item Most promising for me at the moment: sampling + better analysis? But since you mentioned that dependency paths get long, I am also not to sure about that. \\
    (I have to read the papers about that)
    \item Or a completely new idea.
    \item \begin{align}
    \exp\left(-\frac{\Delta^{\log^{-i}\log \Delta}}{\Delta^{\log^{-i+1}\log \Delta}}\right)=\exp\left(-\Delta^{\log^{-i}\log \Delta-\log^{-i+1}\log \Delta})\right)
    \end{align}

    \begin{algorithm}
    \caption{Randomized 2-ruling set for trees}
        \begin{itemize}
            \item[1.] Initialize $S \coloneqq \emptyset$ and $A \coloneqq V$. $W_i \coloneqq \emptyset$ for $i=1, \dots,R.$
            \item[2.] For $i=1, \dots, R$ : 
                \begin{itemize}
                    \item[] For $j=1, \dots, k_i:$ \\
                        $S_j \coloneqq $ \LMJ($G[A])$.\\
                        Set $S \leftarrow S \cup S_j$ and $A \leftarrow A \setminus (S_j \cup N_2(S_j))$.

                \end{itemize}
            $W_i \leftarrow \{ v \in A \mid d_v > \Delta_i \}$ and $B_i \leftarrow W_i \setminus A$. \\
            $A \leftarrow  A \setminus W_i$.
            \item[3.] In parallel for $i=1,\dots ,R$: Post-CleanUp($B_i$).
        
            \item[4.] $M \leftarrow MIS(G[A])$.
            \item[5.] Return $(M \cup S)$.
            
        \end{itemize}
    \end{algorithm}
    
\end{itemize}

\clearpage
\section{Treees: Old proof, ignoriring dependencies}

\begin{theorem}
    There is a randomized distributed algorithm that computes a $\beta$-ruling set for $\Delta$-regular trees (for constant $\beta \geq 3)$ $w.h.p.$ in $O(\log \log n)$ rounds in the LOCAL model.
\end{theorem}

\begin{algorithm}
    \caption{Randomized $\beta$-ruling set}
    \begin{itemize}
        \item[1.] Initialize $S \coloneq \emptyset$ and set all nodes to active and let $A \subseteq V$ denote the set of active nodes.
        \item[2.] Compute $r(v) \in [0,1]$ uniformly at random for all nodes.\\
        \textbf{If} $r(v) < r(w)$ for all $w \in N(v) \cap A$: add $v$ to $S$ and set $v$ and $N(v)$ to inactive.
    \end{itemize}
\end{algorithm}

We denote the $\beta$-hop neighborhood of a node $v$ as $N_{\beta}(v) = \{ w \in V | dist(v,w) \leq \beta \}$ and say that a node $v$ is $covered$, if there exists a node $w \in S \cap N_{\beta}(v)$.

For the first part of the analysis we will assume that the max degree of G is large, i.e. $\Delta \geq \log^{\frac{1}{\beta -2}} n$.\\

\begin{lemma}
    If $\Delta \geq \log^{\frac{1}{\beta -2}} n$ Algorithm 1 computes a $\beta$-ruling set w.h.p. after O(1) rounds.
\end{lemma}

\begin{proof}
In order to prove this lemma consider the probability, that a node $v$ stays uncovered after k round:

\begin{align*}
    \mathbb{P}[\text{$v$ is not covered after $k$ rounds}] \leq ((1- \frac{1}{\Delta + 1})^{\Delta^{\beta-1}})^k\leq e^{-\frac{\Delta^{\beta-1}}{\Delta + 1}\cdot k} \\
    \leq e^{-\frac{\Delta^{\beta-1}}{2 \cdot \Delta}\cdot k} 
    \leq e^{-\Delta^{\beta - 2}\cdot \frac{1}{2k}} \leq n^{-c}.
\end{align*}

    Here we use the fact, that each node is only dependent on the randomness in its $\beta$-hop neighborhood and every node stays active until itself or a direct neighbor gets added to $S$. Thus if $v$ is not covered, all nodes in its ($\beta$-1)-hop neighborhood are still active.
    Since this holds for every node $v \in V$, we know that every node is covered $w.h.p.$ $(1-\frac{1}{n^c}$) after k rounds.    
\end{proof}

    Now let us assume that the max. degree is small, i.e. $\Delta < \log^{\frac{1}{\beta -2}}n$.

    \begin{lemma}
        Let $c > 0$ be an arbitrary constant. For any 2$(\beta+1)$-independent set B the probability that all nodes of B remain uncovered after $\Theta (c \log \log n)$ rounds of Algorithm 1 is at most $\Delta^{-c |B|}$.
    \end{lemma}

    \begin{proof}
        Let $c_1 = 2c$.\\
        By the proof of lemma 1, we know that the probability that a node $v$ stays uncovered after $c_1  \log \log n$ rounds is at most 
        \begin{equation*}
            e^{-\Delta^{\beta - 2}\cdot \frac{c_1}{2} \log \log n} = \log^{-\Delta^{\beta -2}\cdot c}  n
        \end{equation*}
        Since $\beta \geq 3$, it holds that $c \log \log n \geq c\frac{1}{\beta -2} \log \log n  = c \log \log^{\frac{1}{\beta -2 }}  n$.\\
        Thus the fact that $\Delta < \log^{\frac{1}{\beta - 2}} n$ implies that 
        \begin{equation*}
            c \log \log n \geq c   \log \Delta \geq \frac{c}{\Delta^{\beta-2}} \log \Delta
        \end{equation*}

        so we get

        \begin{equation*}
            c \log \log^{\Delta^{\beta -2 }}  n \geq c \log \Delta
        \end{equation*}

        and thus we have 
        \begin{equation*}
            \log ^{-\Delta^{\beta-2} \cdot c } n\leq \Delta^{-c}.
        \end{equation*}

        Since the probability of each node to get covered only depends on its $\beta$-hop neighborhood and a node can only get inactive, if itself or a direct neighbor joins $S$, it follows that each nodes randomness only depends on its $(\beta+1)$-hop neighborhood.\\
        Since each pair of nodes in $B$ has at least distance $2(\beta+1)$ their randomness is independent of each other.\\
        Thus,
        \begin{align*}
            \mathbb{P}[\text{$B$ remains uncovered after $c_1 \log \log n$ rounds}] \\
            = \Pi_{v \in B} \mathbb{P}[\text{$v$ remains uncovered after $c_1 \log \log n$ rounds}] \leq \Pi_{v \in B} \Delta^{-c} = \Delta^{-c|B|}.
        \end{align*}

    \end{proof}

    Now applying the shattering lemma (weiter ausführen?), we obtain the following statement:

    \begin{lemma}
        Let $c>0$ be an arbitrary constant. Let B be the set of uncovered nodes after $\Theta(c \log \log n)$ rounds. Then with probability at least $1-n^{-c_1}$, for a tune-able constant $c_1 = c-2(\beta +1) -3$, any connected component of $G[B]$ has size at most $\log_{\Delta} n \Delta^{2(\beta+1)}$.
    \end{lemma}
    Finish-off: We are now only considering components of size 
    \begin{equation*}
         N =  O(\log_{\Delta}n \cdot \Delta^{4(\beta+1)} )< O(\log_{\Delta}n \cdot \log^{\frac{4(\beta+1)}{2-\beta}} n) = O(\log^{\frac{3\beta + 6}{2- \beta}} n)
        \end{equation*}
    
    Version 1:\\
    Since we are considering trees, we know that the arboricity $a(G) = 1$, thus we can compute a H-partition of size $O(\log N)$ and degree at most $(2+ \epsilon) \cdot 2$ in $O(\log N)$ time ($0 < \epsilon \leq 1$). Now we can compute a valid coloring $W_i$ for each $H_i$ with at most $25 \geq ((2+ \epsilon) \cdot 2)^{2}$ colors ($ i \in [ \ell ]$, for $\ell = O(\log N)$).
    Now we can compute an MIS using this coloring as follows:
    Initialize $S \coloneq W_{\ell 1}$ (first color-class of $W_{\ell}$) and extend it to a valid MIS for $H_{\ell}$, by iterating through all color-classes and checking each node $v$ (in parallel) of that color-class weither $v$ can join $S$. \\
    Now we can extend the MIS to an MIS of the component by iterating through all layers $H_{\ell -1}, \dots , H_1$ and repeating the above procedure (i.e. going through all color-classes).\\
    This procedure takes $O(\log N)$ time, since we have $O(\log N)$ layers and each is done in $25 = O(1)$ rounds.\\
    Thus in total we get a total runtime of $O(\log N + \log N) = O(\log \log^{\frac{3 \beta 
    6}{2- \beta}}n) = O( \frac{3 \beta +6}{2- \beta} \log \log n) = O(\log \log n)$, since $\beta = O(1),$ for computing the H-partition and the MIS.\\

    Version 2:\\

    \clearpage

    If we assume all degrees are in the interval $[ x, x^{1.1}]$, i.e. $d(v) \in [x,x^{1.1} ]$, in particular $\Delta \leq x^{1.1}$.
    Then we have 
    \begin{equation*}
        \mathbb{P}[\text{v is not covered after $O(1)$ rounds}] \leq (1-\frac{1}{x^{1.1}+1})^{x^\beta} \leq e^{-\frac{x^\beta}{x^{1.1}}}.
    \end{equation*}

    We want that the probability that v is not covered is at most $\Delta^{-c}$ for some constant $c\geq 1$ in order to apply shattering.
    But if we let $x= \log \Delta$ then we have:

    \begin{equation*}
        e^{-\frac{x^\beta}{x^{1.1}}} = e^{-\frac{\log ^\beta \Delta}{\log^{1.1} \Delta}} = e^{- \log^{(\beta-1.1)} \Delta} \leq e^{- \log \Delta} = \Delta^{-1}.
    \end{equation*}

    Where the last equation holds for $\beta \geq 3$, since then $(\beta-1.1) \geq 1$.

    If we want to generalize to the intervals $[x^{a},x^{a+0.1})$, we have $\Delta < x^{a+0.1}$ and get 
    
    \begin{equation*}
        \mathbb{P}[\text{v is not covered after $O(1)$ rounds}] \leq (1-\frac{1}{x^{a+0.1}})^{x^{a \cdot \beta}} \leq e^{-\frac{x^{a \cdot \beta}}{x^{a+0.1}}} = e^{-x^{a \cdot (\beta-1)-0.1}} \leq e^{-x^{2a-0.1}}.    
    \end{equation*}

    Where the last inequality holds again for $\beta \geq 3$. If we plug in $x= \log \Delta$:

    \begin{equation*}
        e^{-x^{2a-0.1}} = e^{-\log^{2a-0.1} \Delta} \leq e^{-\log \Delta} = \Delta^{-1}.
    \end{equation*}
    Since $a\geq 1$. (general case also covers first case for $a=1$)

    \clearpage
    \printbibliography    
\end{document}